\documentclass[a4paper]{article}
\usepackage[margin=1in]{geometry}
\usepackage{authblk}
\usepackage[utf8]{inputenc}
\usepackage[T1]{fontenc}
\usepackage{microtype}
\usepackage{amsfonts}
\usepackage{amssymb}
\usepackage{amsthm}
\usepackage{amsmath}
\usepackage{comment}
\usepackage{times}
\usepackage{enumitem}
\usepackage[hidelinks]{hyperref}
\usepackage[capitalise]{cleveref}
\usepackage{thmtools}
\usepackage{booktabs}
\usepackage[bf,small]{caption}

\newtheorem{theorem}{Theorem}[section]
\newtheorem{lemma}[theorem]{Lemma}

\newtheorem{corollary}[theorem]{Corollary}
\newtheorem{definition}[theorem]{Definition}

\theoremstyle{remark}
\newtheorem{example}[theorem]{Example}
\newcommand{\bigO}{\mathcal{O}}

\let\OLDthebibliography\thebibliography
  \renewcommand\thebibliography[1]{
  \OLDthebibliography{#1}
  \setlength{\parskip}{0pt}
  \setlength{\itemsep}{0pt plus 0.3ex}
}

\begin{document}

\setenumerate{itemsep=0ex, parsep=1pt}

\renewcommand\Affilfont{\normalsize}

\title{At the Roots of Dictionary Compression: String Attractors}
\author[1]{Dominik Kempa}
\author[2,3]{Nicola Prezza}

\affil[1]{Department of Computer Science, University of Helsinki, Finland}
\affil[ ]{\href{mailto:dkempa@cs.helsinki.fi}{\nolinkurl{dkempa@cs.helsinki.fi}}}
\affil[2]{DTU Compute, Technical University of Denmark, Denmark}
\affil[3]{Department of Computer Science, University of Pisa, Italy}
\affil[ ]{\href{nicola.prezza@di.unipi.it}{\nolinkurl{nicola.prezza@di.unipi.it}}}

\date{\vspace{-1.5cm}}
\maketitle

\begin{abstract}
  A well-known fact in the field of lossless text compression is that
  high-order entropy is a weak model when the input contains long
  repetitions. Motivated by this fact, decades of research have
  generated myriads of so-called \emph{dictionary compressors}:
  algorithms able to reduce the text's size by exploiting its
  repetitiveness. Lempel-Ziv 77 is one of the most successful and
  well-known tools of this kind, followed by straight-line programs,
  run-length Burrows-Wheeler transform, macro schemes, collage
  systems, and the compact directed acyclic word graph. In this paper,
  we show that these techniques are different solutions to the same,
  elegant, combinatorial problem: to find a small set of positions
  capturing all distinct text's substrings. We call such a set a
  \emph{string attractor}.  We first show reductions between
  dictionary compressors and string attractors. This gives the
  approximation ratios of dictionary compressors with respect to the
  smallest string attractor and allows us to uncover new asymptotic
  relations between the output sizes of different dictionary
  compressors. We then show that the \emph{$k$-attractor problem} ---
  deciding whether a text has a size-$t$ set of positions capturing
  all substrings of length at most $k$ --- is NP-complete for $k\geq
  3$. This, in particular, includes the full string attractor
  problem. We provide several approximation techniques for the
  smallest $k$-attractor, show that the problem is APX-complete for
  constant $k$, and give strong inapproximability results.  To
  conclude, we provide matching lower and upper bounds for the random
  access problem on string attractors.  The upper bound is proved by
  showing a data structure supporting queries in optimal time. Our
  data structure is \emph{universal}: by our reductions to string
  attractors, it supports random access on any dictionary-compression
  scheme. In particular, it matches the lower bound also on LZ77,
  straight-line programs, collage systems, and macro schemes, and
  therefore essentially closes (at once) the random access problem for
  all these compressors.
\end{abstract}

\section{Introduction}

The goal of lossless text compression is to reduce the size of a given
string by exploiting irregularities such as skewed character
distributions or substring repetitions.  Unfortunately, the holy grail
of compression --- Kolmogorov complexity~\cite{kolmogorov1965three}
--- is non-computable: no Turing machine can decide, in a finite
number of steps, whether a given string has a program generating it
whose description is smaller than some fixed value $K$. This fact
stands as the basis of all work underlying the field of data
compression: since we cannot always achieve the best theoretical
compression, we can at least try to approximate it. In order to
achieve such a goal, we must first find a model that captures, to some
good extent, the degree of regularity of the text. For example, in the
case of the text generated by a Markovian process of order $k$, the
$k$-th order entropy $H_k$ of the source represents a lower bound for
our ability to compress its outputs. This concept can be extended to
that of \emph{empirical entropy}~\cite{cover2006elements} when the
underlying probabilities are unknown and must be approximated with the
empirical symbol frequencies.  The $k$-th order compression, however,
stops being a reasonable model about the time when $\sigma^k$ becomes
larger than $n$, where $\sigma$ and $n$ are the alphabet size and the
string length, respectively.  In particular,
Gagie~\cite{gagie2006large} showed that when $k \geq \log_\sigma n$,
no compressed representation can achieve a worst-case space bound of
$c \cdot n H_k + o (n \log \sigma)$ bits, regardless of the value of
the constant $c$.  This implies that $k$-th order entropy is a weak
model when $k$ is large, i.e., when the goal is to capture long
repetitions. Another way of proving this fact is to observe that, for
any sufficiently long text $T$, symbol frequencies (after taking their
context into account) in any power of $T$ (i.e., $T$ concatenated with
itself) do not vary
significantly~\cite[Lem. 2.6]{kreft2013compressing}. As a result, we
have that $t\cdot nH_k(T^t) \approx t\cdot nH_k(T)$ for any $t>1$: entropy is
not sensitive to very long repetitions.

This particular weakness of entropy compression generated, in the last
couple of decades, a lot of interest in algorithms able to directly
exploit text repetitiveness in order to beat the entropy lower bound
on very repetitive texts.  The main idea underlying these algorithms
is to replace text substrings with references to a dictionary of
strings, hence the name \emph{dictionary compressors}.  One effective
compression strategy of this kind is to build a context-free grammar
that generates (only) the string. Such grammars (in Chomsky normal
form) are known by the name of \emph{straight-line programs}
(SLP)~\cite{KiefferY00}; an SLP is a set of rules of the kind
$X\rightarrow AB$ or $X \rightarrow a$, where $X$, $A$, and $B$ are
\emph{nonterminals} and $a$ is a \emph{terminal}. The string is
obtained from the expansion of a single starting nonterminal $S$. If
also rules of the form $X \rightarrow A^\ell$ are allowed, for any
$\ell > 2$, then the grammar is called \emph{run-length SLP}
(RLSLP)~\cite{NishimotoIIBT16}.  The problems of finding the smallest
SLP --- of size $g^*$ --- and the smallest run-length SLP --- of size
$g_{rl}^*$ --- are NP-hard~\cite{CLLPPSS05, hucke2016smallest}, but
fast and effective approximation algorithms are known, e.g.,
LZ78~\cite{ziv1978compression}, LZW~\cite{welch1984technique},
Re-Pair~\cite{larsson2000off}, Bisection~\cite{kieffer2000universal}.
An even more powerful generalization of RLSLPs is represented by
\emph{collage systems}~\cite{KidaMSTSA03}: in this case, also rules of
the form $X\rightarrow Y[l..r]$ are allowed (i.e., $X$ expands to a
substring of $Y$). We denote with $c$ the size of a generic collage
system, and with $c^*$ the size of the smallest one.  A related
strategy, more powerful than grammar compression, is that of replacing
repetitions with pointers to other locations in the string. The most
powerful and general scheme falling into this category takes the name
of \emph{pointer macro scheme}~\cite{storer1978macro,storer1982data},
and consists of a set of substring equalities that allow for
unambiguously reconstructing the string. Finding the smallest such
system --- of size $b^*$ --- is also
NP-hard~\cite{gallant1982string}. However, if we add the constraint of
unidirectionality (i.e., text can only be copied from previous
positions), then Lempel and Ziv in~\cite{lempel1976complexity} showed
that a greedy algorithm (LZ77) finds an optimal solution to the
problem (we denote the size of the resulting parsing by $z$).
Subsequent works showed that LZ77 can even be computed in linear
time~\cite{crochemore2008computing}.  Other effective techniques to
compress repetitive strings include the run-length Burrows-Wheeler
transform~\cite{burrows1994block} (RLBWT) and the compact directed
acyclic word graph~\cite{blumer1987complete,crochemore1997direct}
(CDAWG). With the first technique, we sort all circular string
permutations in an $n\times n$ matrix; the BWT is the last column of
this matrix. The BWT contains few equal-letter runs if the string is
very repetitive, therefore run-length compression often significantly
reduces the size of this string
permutation~\cite{makinen2010storage}. The number $r$ of runs in the
BWT is yet another good measure of repetitiveness. Finally, one can
build a compact (that is, path compressed) automaton recognizing the
string's suffixes, and indicate with $e$ the number of its edges. The
size $e^*$ of the smallest such automaton --- the CDAWG --- also grows
sublinearly with $n$ if the string is very
repetitive~\cite{belazzougui2015composite}. Both RLBWT and CDAWG can
be computed in linear
time~\cite{munro2016space,belazzougui2014linear,crochemore1997compact}.

The promising results obtained in the field of dictionary compression
have generated --- in recent years --- a lot of interest around the
closely-related field of \emph{compressed computation}, i.e.,
designing compressed data structures that efficiently support a
particular set of queries on the text. The sizes of these data
structures are bounded in terms of repetitiveness measures. As with
text compression, the landscape of compressed data structures is
extremely fragmented: different solutions exist for each compression
scheme, and their space/query times are often not even comparable, due
to the fact that many asymptotic relations between repetitiveness
measures are still missing. See, for example, Gagie et
al.~\cite{gagie2017optimal} for a comprehensive overview of the
state-of-the-art of dictionary-compressed full-text indexes (where
considered queries are random access to text and counting/locating
pattern occurrences). In this paper we consider data structures
supporting random access queries (that is, efficient local
decompression).  Several data structures for this problem have been
proposed in the literature for each distinct compression scheme. In
Table \ref{table:extract} we report the best time-space trade-offs
known to date (grouped by compression scheme).  Extracting text from
Lempel-Ziv compressed text is a notoriously difficult problem.  No
efficient solution is known within $\bigO(z)$ space (they all require
time proportional to the parse's height), although efficient queries
can be supported by raising the space by a logarithmic
factor~\cite{PhBiCPM17,BGGKOPT15}. Grammars, on the other hand, allow
for more compact and time-efficient extraction strategies.  Bille et
al.~\cite{BLRSRW15} have been the first to show how to efficiently
perform text extraction within $\bigO(g)$ space. Their time bounds
were later improved by Belazzougui et al.~\cite{BPT15}, who also
showed how to slightly increase the space to $\bigO(g\log^\epsilon
n\log(n/g))$ while matching a lower bound holding on
grammars~\cite{CVY13}. Space-efficient text extraction from the
run-length Burrows-Wheeler transform has been an open problem until
recently. Standard solutions~\cite{makinen2010storage} required
spending additional $\bigO(n/s)$ space on top of the RLBWT in order to
support extraction in a time proportional to $s$. In a recent
publication, Gagie et al.~\cite{gagie2017optimal} showed how to
achieve near-optimal extraction time in the packed setting within
$\bigO(r\log(n/r))$ space. Belazzougui and Cunial~\cite{BCspire17}
showed how to efficiently extract text from a CDAWG-compressed
text. Their most recent work~\cite{belazzougui2017representing} shows,
moreover, how to build a grammar of size $\bigO(e)$: this result
implies that the solutions for grammar-compressed text can be used on
the CDAWG. To conclude, no strategies for efficiently extracting text
from general macro schemes and collage systems are known to date: the
only solution we are aware of requires explicitly navigating the
compressed structure, and can therefore take time proportional to the
text's length in the worst case.

\paragraph{Our Contributions}

At this point, it is natural to ask whether there exists a common (and
simple) principle underlying the complex set of techniques
constituting the fields of dictionary compression and
compressed-computation.  In this paper, we answer (affirmatively) this
question. Starting from the observation that string repetitiveness can
be defined in terms of the cardinality of the set of distinct
substrings, we introduce a very simple combinatorial object --- the
\emph{string attractor} --- capturing the complexity of this
set. Formally, a string attractor is a subset of the string's
positions such that all distinct substrings have an occurrence
crossing at least one of the attractor's elements. Despite the
simplicity of this definition, we show that dictionary compressors can
be interpreted as algorithms approximating the smallest string
attractor: they all induce (very naturally) string attractors whose
sizes are bounded by their associated repetitiveness measures.  We
also provide reductions from string attractors to most dictionary
compressors and use these reductions to derive their approximation
rates with respect to the smallest string attractor.  This yields our
first efficient approximation algorithms computing the smallest string
attractor, and allows us to uncover new relations between
repetitiveness measures. For example, we show that $g^*,z \in
\bigO(c^*\log^2(n/c^*))$, $c^* \in
\bigO(b^*\log(n/b^*))\ \cap\ \bigO(r\log(n/r))$, and $b^* \in
\bigO(c^*\log(n/c^*))$.

Our reductions suggest that a solution (or a good approximation) to
the problem of finding an attractor of minimum size could yield a
better understanding of the concept of text repetitiveness and could
help in designing better dictionary compressors. We approach the
problem by first generalizing the notion of string attractor to that
of $k$-attractor: a subset of the string's positions capturing all
substrings of length at most $k$. We study the computational
complexity of the $k$-attractor problem: to decide whether a text has
a $k$-attractor of a given size. Using a reduction from $k$-set-cover,
we show that the $k$-attractor problem is NP-complete for $k\geq
3$. In particular, this proves the NP-completeness of the original
string attractor problem (i.e., the case $k=n$). Given the hardness of
computing the smallest attractor, we focus on the problem of
approximability. We show that the smallest $k$-attractor problem is
APX-complete for constant $k$ by showing a $2k$-approximation
computable in linear time and a reduction from $k$-vertex-cover. We
also use reductions to $k$-set-cover to provide $\bigO(\log
k)$-approximations computable in polynomial time.  Our
APX-completeness result implies that the smallest $k$-attractor
problem has no PTAS unless P=NP.  Using a reduction from
$3$-vertex-cover and explicit constants derived by Berman and
Karpinski~\cite{bk99}, we strengthen this result and show that, for
every $\epsilon > 0$ and every $k \geq 3$, it is NP-hard to
approximate the smallest $k$-attractor within a factor of 11809/11808
- $\epsilon$.

We proceed by presenting an application of string attractors to the
domain of compressed computation: we show that the simple property
defining string attractors is sufficient to support random access in
optimal time.  We first extend a lower bound~\cite[Thm. 5]{CVY13} for
random access on grammars to string attractors.  Let $\gamma$ be the
size of a string attractor of a length-$n$ string $T$ over an alphabet
of size $\sigma$.  The lower bound states that $\Omega(\log n/\log\log
n)$ time is needed to access one random position within $\bigO(\gamma
\ \mathrm{polylog}\ n)$ space.  Let $w$ be the memory word size (in
bits). We present a data structure taking $\bigO(\gamma \tau
\log_\tau(n/\gamma))$ words of space and supporting the extraction of
any length-$\ell$ substring of $T$ in $\bigO(\log_\tau(n/\gamma) +
\ell\log\sigma/w)$ time, for any $\tau\geq 2$ fixed at construction
time.  For $\tau=\log^\epsilon n$ (for any constant $\epsilon>0$) this
query time matches the lower bound.  Choosing $\tau =
(n/\gamma)^{\epsilon}$, we obtain instead optimal time in the packed
setting within $\bigO(\gamma^{1-\epsilon}n^\epsilon)$ space.  From our
reductions, our solution is \emph{universal}: given a
dictionary-compressed text representation, we can induce a string
attractor of the same size and build our structure on top of it.  We
note, as well, that the lower bound holds, in particular, on most
compression schemes.  As a result, our data structure is also optimal
for SLPs, RLSLPs, collage systems, LZ77, and macro schemes.  Tables
\ref{table:extract} and \ref{table:extract A-DAG} put our structure in
the context of state-of-the-art solutions to the problem. Note that
all existing solutions depend on a specific compression scheme.

\section{Preliminaries}\label{sec:preliminaries}

Throughout the paper, we use the terms \emph{string} and \emph{text}
interchangeably. The notion $T[i..j]$, $1 \leq i \leq j \leq n$,
denotes the substring of string $T\in\Sigma^n$ starting at position
$i$ and ending at position $j$. We denote the alphabet size of string
$T$ by $|\Sigma|=\sigma$.

\begin{table}
  \centering
  \caption{Best trade-offs in the literature for extracting text from
    compressed representations.}
  \begin{tabular}{l c c}
    \toprule
    Structure & Space & Extract time\\
    \midrule
    \cite[Lem.~5]{PhBiCPM17}
    & $\bigO(z\log(n/z))$ & $\bigO(\ell+\log(n/z))$\\
    \cite[Thm.~2]{BGGKOPT15}   
    & $\bigO(z\log(n/z))$ & $\bigO((1+\ell/\log_\sigma n)\log(n/z))$\\
    \cite[Thm.~1]{BPT15} 
    & $\bigO(g)$ &  $\bigO(\ell/\log_\sigma n+\log n)$ \\
    \cite[Thm.~3]{BPT15}
    & $\bigO(g\log^\epsilon n\log\frac{n}{g})$ &
    $\bigO(\ell/\log_{\sigma}n+\frac{\log n}{\log \log n})$\\
    \cite[Thm.~4]{gagie2017optimal}
    & $\bigO(r\log(n/r))$ & $\bigO(\ell\log(\sigma)/w + \log(n/r))$ \\
    \cite[Thm.~7]{belazzougui2017representing}
    & $\bigO(e)$ & $\bigO(\ell/\log_{\sigma}n + \log n)$\\
    \bottomrule
  \end{tabular}
  \label{table:extract}
\end{table}

\begin{table}
  \centering
  \caption{Some trade-offs achievable with our structure for different
    choices of $\tau$, in order of decreasing space and increasing
    time. Query time in the first row is optimal in the packed
    setting, while in the second row it is optimal within the
    resulting space due to a lower bound for random access on string
    attractors. To compare these bounds with those of Table
    \ref{table:extract}, just replace $\gamma$ with any of the
    measures $z$, $g$, $r$, or $e$ (possible by our reductions to
    string attractors).}
  \begin{tabular}{ccc}
    \toprule
    $\tau$ & Space & Extract time \\
    \midrule
    $(n/\gamma)^{\epsilon}$&$\bigO(\gamma^{1-\epsilon}n^\epsilon)$ & $\bigO(\ell\log(\sigma)/w)$\\
    $\log^\epsilon n$ &$\bigO\left(\gamma \log^{\epsilon} n\log(n/\gamma)\right)$ & 
    $\bigO\left(\ell\log(\sigma)/w + \frac{\log(n/\gamma)}{\log\log n}\right)$ \\
    2 &$\bigO\left(\gamma \log(n/\gamma)\right)$ & $\bigO(\ell\log(\sigma)/w + \log(n/\gamma))$\\
    \bottomrule
  \end{tabular}
  \label{table:extract A-DAG}
\end{table}

The \emph{LZ77 parsing}~\cite{LZ77,lempel1976complexity} of a string
$T$ is a greedy, left-to-right parsing of $T$ into \emph{longest
  previous factors}, where a longest previous factor at position $i$
is a pair $(p_i,\ell_i)$ such that, $p_i < i$, $T[p_i..p_i+\ell_i-1] =
T[i..i+\ell_i-1]$, and $\ell_i$ is maximized.  In this paper, we use
the LZ77 variant where no overlaps between phrases and sources are
allowed, i.e., we require that $p_i+\ell_i-1 < i$.  Elements of the
parsing are called \emph{phrases}.  When $\ell_i > 0$, the substring
$T[p_i..p_i+\ell_i-1]$ is called the {\em source} of phrase
$T[i..i+\ell_i-1]$.  In other words, $T[i..i+\ell_i-1]$ is the longest
prefix of $T[i..n]$ that has another occurrence not overlapping
$T[i..n]$ and $p_i<i$ is its starting position.  The exception is when
$\ell_i=0$, which happens iff $T[i]$ is the leftmost occurrence of a
symbol in $T$. In this case we output $(T[i],0)$ (to represent
$T[i..i]$: a phrase of length 1) and the next phrase starts at
position $i+1$.  LZ77 parsing has been shown to be the smallest
parsing of the string into phrases with sources that appear earlier in
the text~\cite{LZ77}. The parsing can be computed in $\bigO(n)$ time
for integer alphabets~\cite{crochemore2008computing} and in $\bigO(n
\log \sigma)$ for general (ordered) alphabets~\cite{RodehPE81}.  The
number of phrases in the LZ77 parsing of string $T$ is denoted by $z$.

A \emph{macro scheme}~\cite{storer1982data} is a set of $b$ directives
of two possible types:
\begin{enumerate}
  \item $T[i..j] \leftarrow T[i'..j']$ (i.e., copy $T[i'..j']$ in
    $T[i..j]$), or
  \item $T[i] \leftarrow c$, with $c\in\Sigma$ (i.e., assign character
    $c$ to $T[i]$),
\end{enumerate}
such that $T$ can be unambiguously decoded from the directives.

A \emph{bidirectional parse} is a macro scheme where the left-hand
sides of the directives induce a text factorization, i.e., they cover
the whole $T$ and they do not overlap.  Note that LZ77 is a particular
case of a bidirectional parse (the optimal unidirectional one), and
therefore it is also a macro scheme.

A \emph{collage system}~\cite{KidaMSTSA03} is a set of $c$ rules of
four possible types:
\begin{enumerate}
  \item $X \rightarrow a$: nonterminal $X$ expands to a terminal $a$,
  \item $X \rightarrow AB$: nonterminal $X$ expands to $AB$, with $A$
    and $B$ nonterminals different from $X$,
  \item $X \rightarrow R^{\ell}$: nonterminal $X$ expands to
    nonterminal $R\neq X$ repeated $\ell$ times,
  \item $X \rightarrow K[l..r]$: nonterminal $X$ expands to a
    substring of the expansion of nonterminal $K\neq X$.
\end{enumerate}
The text is the result of the expansion of a special starting
nonterminal $S$. Moreover, we require that the collage system does not
have cycles, i.e., the derivation tree of any nonterminal $X$ does not
contain $X$ nor $X[l..r]$ for some integers $l,r$ as an internal node.
Collage systems generalize SLPs (where only rules 1 and 2 are allowed)
and RLSLPs (where only rules 1, 2, and 3 are allowed). The
\emph{height} $h_X$ of a nonterminal $X$ is defined as follows. If $X$
expands to a terminal with rule 1, then $h_X=1$. If $X$ expands to
$AB$ with rule 2, then $h_X = \max\{h_A, h_B\}+1$. If $X$ expands to
$R^\ell$ with rule 3, then $h_X = h_R+1$. If $X$ expands to $K[l..r]$
with rule 4, then $h_X = h_K+1$. The \emph{height of the collage
  system} is the height of its starting nonterminal.

By $SA[1..n]$ we denote the \emph{suffix array}~\cite{ManberM93} of
$T$, $|T|=n$, defined as a permutation of the integers $[1..n]$ such
that $T[SA[1]..n] \prec T[SA[2]..n] \prec \cdots \prec T[SA[n]..n]$,
where $\prec$ denotes the lexicographical ordering.  For simplicity we
assume that $T[n]=\$$, where $\$$ is a special symbol not occurring
elsewhere in $T$ and lexicographically smaller than all other alphabet
symbols.  The \emph{Burrows-Wheeler Transform}~\cite{burrows1994block}
$BWT[1..n]$ is a permutation of the symbols in $T$ such that $BWT[i] =
T[SA[i]-1]$ if $SA[i]>1$ and $\$$ otherwise.  Equivalently, BWT can be
obtained as follows: sort lexicographically all cyclic permutations of
$T$ into a (conceptual) matrix $M\in\Sigma^{n\times n}$ and take its
last column. Denote the first and last column of $M$ by $F$ and $L$,
respectively. The key property of $M$ is the \emph{LF mapping}: the
$i$-th occurrence of any character $c$ in column $L$ corresponds to
the $i$-th occurrence of any character $c$ in column $F$ (i.e., they
represent the same position in the text). With $LF[i]$, $i\in[1,n]$ we
denote the LF mapping applied on position $i$ in the $L$ column. It is
easy to show that $LF[i]=C[L[i]]+{\rm rank}_{L[i]}(L, i)$, where
$C[c]=|\{i\in [1,n] \mid L[i]<c\}|$ and ${\rm rank}_c(L, i)$ is the
number of occurrences of $c$ in $L[1..i-1]$.

On compressible texts, BWT exhibits some remarkable properties that
allow the boosting of compression. In particular, it can be
shown~\cite{makinen2010storage} that repetitions in $T$ generate
equal-letter runs in BWT. We can efficiently represent this transform
as the list of pairs
$$
RLBWT = \langle  \lambda_i, c_i \rangle_{i=1,\dots, r},
$$ where $\lambda_i>0$ is the length of the $i$-th maximal run, and
$c_i\in\Sigma$. Equivalently, $RLBWT$ is the shortest list of pairs
$\langle \lambda_i, c_i \rangle_{i=1,\dots, r}$ satisfying $BWT =
c_1^{\lambda_1}c_2^{\lambda_2}\dots c_{r}^{\lambda_{r}}$.

The \emph{compact directed acyclic word
  graph}~\cite{blumer1987complete,crochemore1997direct} (CDAWG for
short) is the minimum path-compressed graph (i.e., unary paths are
collapsed into one path) with one source node $s$ and one sink node
$f$ such that all $T$'s suffixes can be read on a path starting from
the source. The CDAWG can be built in linear time by minimization of
the suffix tree~\cite{Weiner73} of $T$: collapse all leaves in one
single node, and proceed bottom-up until no more nodes of the suffix
tree can be collapsed. The CDAWG can be regarded as an automaton
recognizing all $T$'s substrings: make $s$ the initial automaton's
state and all other nodes (implicit and explicit) final.

\section{String Attractors}\label{sec:SA}

A string $T[1..n]$ is considered to be repetitive when the cardinality
of the set $SUB_T = \{ T[i..j]\ |\ 1\leq i\leq j \leq |T| \}$ of its
distinct substrings is much smaller than the maximum number of
distinct substrings that could appear in a string of the same length
on the same alphabet.  Note that $T$ can be viewed as a compact
representation of $SUB_T$. This observation suggests a simple way of
capturing the degree of repetitiveness of $T$, i.e., the cardinality
of $SUB_T$. We can define a function $\phi: SUB_T \rightarrow \Gamma
\subseteq [1,n]$ satisfying the following property: each $s\in SUB_T$
has an occurrence crossing position $\phi(s)$ in $T$. Note that such a
function is not necessarily unique. The codomain $\Gamma$ of $\phi$ is
the object of study of this paper. We call this set a \emph{string
  attractor}:

\begin{definition}\label{def: string attractor}
  A \emph{string attractor} of a string $T\in\Sigma^n$ is a set of
  $\gamma$ positions $\Gamma = \{j_1, \dots, j_\gamma\}$ such that
  every substring $T[i..j]$ has an occurrence $T[i'..j'] = T[i..j]$
  with $j_k \in [i',j']$, for some $j_k\in\Gamma$.
\end{definition}

\begin{example}
  Note that $\{1, 2, .., n\}$ is always a string attractor (the
  largest one) for any string. Note also that this is the only
  possible string attractor for a string composed of $n$ distinct
  characters (e.g., a permutation).
\end{example}

\begin{example}\label{ex:attr}
  Consider the following string $T$, where we underlined the positions
  of a smallest string attractor $\Gamma^* = \{4,7,11,12\}$ of $T$.
  \begin{center}
    \texttt{CDA\underline BCC\underline DABC\underline C\underline A}
  \end{center}
  To see that $\Gamma^*$ is a valid attractor, note that every
  substring between attractor's positions has an occurrence crossing
  some attractor's position (these substrings are \texttt{A},
  \texttt{B}, \texttt{C}, \texttt{D}, \texttt{CD}, \texttt{DA},
  \texttt{CC}, \texttt{AB}, \texttt{BC}, \texttt{CDA},
  \texttt{ABC}). The remaining substrings cross an attractor's
  position by definition.  To see that $\Gamma^*$ is of minimum size,
  note that the alphabet size is $\sigma = 4 = |\Gamma^*|$, and any
  attractor $\Gamma$ must satisfy $|\Gamma| \geq \sigma$.
\end{example}

\subsection{Reductions from Dictionary Compressors}
\label{sec:compressors->SA}

In this section we show that dictionary compressors induce string
attractors whose sizes are bounded by their associated repetitiveness
measures.

Since SLPs and RLSLPs are particular cases of collage systems, we only
need to show a reduction from collage systems to string attractors to
capture these three classes of dictionary compressors.  We start with
the following auxiliary lemma.

\begin{lemma}\label{lem:collage_sys}
  Let $C = 
    \{X_i \rightarrow a_i,\ i=1,\dots, c'\}\cup 
    \{Y_i \rightarrow A_iB_i,\ i=1,\dots,c''\} \cup 
    \{Z_i \rightarrow R_i^{\ell_i},\ \ell_i> 2,\  i=1,\dots,c'''\} \cup 
    \{W_i \rightarrow K_i[l_i..r_i],\ i=1,\dots, c''''\}$ 
    be a collage system with starting nonterminal $S$ generating
    string $T$. For any substring $T[i..j]$ one of the following is true:
  \begin{enumerate}
    \item $i=j$ and $T[i] = a_k$, for some $1\leq k\leq c'$, or
    \item there exists a rule $Y_k \rightarrow A_kB_k$ such that
      $T[i..j]$ is composed of a non-empty suffix of the expansion of
      $A_k$ followed by a non-empty prefix of the expansion of $B_k$,
      or
    \item there exists a rule $Z_k \rightarrow R_k^{\ell_k}$ such that
      $T[i..j]$ is composed of a non-empty suffix of the expansion of
      $R_k$ followed by a non-empty prefix of the expansion of
      $R_k^{\ell_k-1}$.
  \end{enumerate}
\end{lemma}
\begin{proof}
Consider any substring $T[i..j]$ generated by expanding the start rule
$S$. The proof is by induction on the height $h$ of $S$.

For $h=1$, the start rule $S$ must expand to a single symbol and hence
case (1) holds.  Consider a collage system of height $h>1$, and let
$S$ be its start symbol. Then, $S$ has one of the following forms:
\begin{enumerate}
  \item $S$ expands as $S \rightarrow AB$, or
  \item $S$ expands as $S \rightarrow R^\ell$, or
  \item $S$ expands as $S \rightarrow K[l..r]$,
\end{enumerate}
where $A$, $B$, $R$, and $K$ are all nonterminals of height $h-1$.

In case (1), either $T[i..j]$ is fully contained in the expansion of
$A$, or it is fully contained in the expansion of $B$, or it is formed
by a non-empty suffix of the expansion of $A$ followed by a non-empty
prefix of the expansion of $B$. In the first two cases, our claim is
true by inductive hypothesis on the collage systems with start symbols
$A$ or $B$. In the third case, our claim is true by definition.

In case (2), either $T[i..j]$ is fully contained in the expansion of
$R$, or it is formed by a non-empty suffix of the expansion of
$R^{\ell_1}$ followed by a non-empty prefix of the expansion of
$R^{\ell_2}$, for some $\ell_1, \ell_2>0$ such that $\ell_1+ \ell_2 =
\ell$. In the former case, our claim is true by inductive
hypothesis. In the latter case, $T[i..j]$ can be written as a suffix
of $R$ followed by a concatenation of $k\geq 0$ copies of $R$ followed
by a prefix of $R$, i.e., $T[i..j] = R[l..|R|]R^{k}R[1..r]$ for some
$1 \leq l \leq |R|$, $k\geq 0$, and $1 \leq r \leq |R|$. Then,
$T[i..j]$ has also an occurrence crossing $R$ and $R^{\ell-1}$:
$T[i..j] = R[l..|R|]R^{\ell-1}[1..(j-i)-(|R|-l)]$.

In case (3), since $T[i..j]$ is a substring of the expansion of $S$,
then it is also a substring of the expansion of $K$. Since the height
of $K$ is $h-1$, we apply an inductive hypothesis with start symbol
$K$.
\end{proof}

The above lemma leads to our first reduction:

\begin{theorem}\label{th:attr_c}
  Let $C$ be a collage system of size $c$ generating $T$. Then, $T$
  has an attractor of size at most $c$.
\end{theorem}
\begin{proof}
  Let $C = \{X_i \rightarrow a_i,\ i=1,\dots, c'\}\cup \{Y_i
  \rightarrow A_iB_i,\ i=1,\dots,c''\} \cup \{Z_i \rightarrow
  R_i^{\ell_i},\ \ell_i> 2,\ i=1,\dots,c'''\} \cup \{W_i \rightarrow
  K_i[l_i..r_i],\ i=1,\dots, c''''\}$ be a collage system of size
  $c=c'+c''+c'''+c''''$ generating string $T$.
	
  Start with an empty string attractor $\Gamma_C=\emptyset$ and repeat
  the following for $k=1,\ldots,c'$. Choose any of the occurrences
  $T[i]$ of the expansion $a_k$ of $X_k$ in $T$ and insert $i$ in
  $\Gamma_C$.  Next, for $k=1,\ldots,c''$, choose any of the
  occurrences $T[i..j]$ of the expansion of $Y_k$. By the production
  $Y_k \rightarrow A_kB_k$, $T[i..j]$ can be factored as
  $T[i..j]=T[i..i']T[i'+1..j]$, where $T[i..i']$ and $T[i'+1..j]$ are
  expansions of $A_k$ and $B_k$, respectively. Insert position $i'$ in
  $\Gamma_C$.  Finally, for $k=1,\ldots,c'''$, choose any of the
  occurrences $T[i..j]$ of the expansion of $Z_k$ in $T$. By the
  production $Z_k \rightarrow R_k^{\ell_k}$, $T[i..j]$ can be factored
  as $T[i..j]=T[i..i']T[i'+1..j]$, where $T[i..i']$ and $T[i'+1..j]$
  are expansions of $R_k$ and $R_k^{\ell_k-1}$. Insert position $i'$
  in $\Gamma_C$.
	
  Clearly the size of $\Gamma_C$ is at most $c$. To see that
  $\Gamma_C$ is a valid attractor, consider any substring $T[i..j]$ of
  $T$. By Lemma~\ref{lem:collage_sys}, either $i=j$ (and, by
  construction, $\Gamma_C$ contains a position of some occurrence of
  $a_k$ such that $T[i]=a_k$ and $X_k \rightarrow a_k$ is one of the
  rules in $C$), or $T[i..j]$ spans the expansion of some $A_k|B_k$ or
  $R_k|R_k^{\ell_k-1}$ (with the crossing point shown). From the
  construction of $\Gamma_C$, such expansion has an occurrence in $T$
  containing an element in $\Gamma_C$ right before the crossing
  point. Thus, $T[i..j]$ has an occurrence $T[i'..j']$ containing a
  position from $\Gamma_C$.
\end{proof}

We now show an analogous result for macro schemes.

\begin{theorem}\label{th: attractor MS}
  Let M be a macro scheme of size $b$ of $T$. Then, $T$ has an
  attractor of size at most $2b$.
\end{theorem}
\begin{proof}	
  Let $T[i_{k_1}..j_{k_1}] \leftarrow
  T[i'_{k_1}..j'_{k_1}],\ T[q_{k_2}] \leftarrow c_{k_2}$, with $1\leq
  k_1 \leq b_1$, $1\leq k_2 \leq b_2$, and $b=b_1+b_2$ be the $b$
  directives of our macro scheme MS.  We claim that $\Gamma_{MS} =
  \{i_1, \dots, i_{b_1}, j_1, \dots, j_{b_1}, q_1, \dots, q_{b_2}\}$
  is a valid string attractor for $T$.
	
  Let $T[i..j]$ be any substring. All we need to show is that
  $T[i..j]$ has a \emph{primary occurrence}, i.e., an occurrence
  containing one of the positions $i_{k_1}$, $j_{k_1}$, or $q_{k_2}$.
  Let $s_1=i$ and $t_1 = j$. Consider all possible chains of copies
  (following the macro scheme directives) $T[s_1..t_1] \leftarrow
  T[s_2..t_2] \leftarrow T[s_3..t_3] \leftarrow \dots$ that either end
  in some primary occurrence $T[s_k..t_k]$ or are infinite (note that
  there could exist multiple chains of this kind since the left-hand
  side terms of some macro scheme's directives could overlap).  Our
  goal is to show that there must exist at least one finite such
  chain, i.e., that ends in a primary occurrence. Pick any $s_1\leq
  p_1 \leq t_1$. Since ours is a valid macro scheme, then $T[p_1]$ can
  be retrieved from the scheme, i.e., the directives induce a finite
  chain of copies $T[p_1] \leftarrow \dots \leftarrow T[p_{k'}]
  \leftarrow c$, for some $k'$, such that $T[p_{k'}] \leftarrow c$ is
  one of the macro scheme's directives. We now show how to build a
  finite chain of copies $T[s_1..t_1] \leftarrow T[s_2..t_2]
  \leftarrow \dots \leftarrow T[s_k..t_k]$ ending in a primary
  occurrence $T[s_k..t_k]$ of $T[s_1..t_1]$, with $k \leq k'$. By
  definition, the assignment $T[p_1] \leftarrow T[p_2]$ comes from
  some macro scheme's directive $T[l_1..r_1] \leftarrow T[l_2..r_2]$
  such that $p_1 \in [l_1,r_1]$ and $p_1-l_1 = p_2-l_2$ (if there are
  multiple directives of this kind, pick any of them). If either
  $l_1\in [s_1,t_1]$ or $r_1\in [s_1,t_1]$, then $T[s_1..t_1]$ is a
  primary occurrence and we are done. Otherwise, we set $s_2 = l_2 +
  (i-l_1)$ and $t_2 = l_2 + (j-l_1)$. By this definition, we have that
  $T[s_1..t_1] = T[s_2..t_2]$ and $p_2\in [s_2,t_2]$, therefore we can
  extend our chain to $T[s_1..t_1] \leftarrow T[s_2..t_2]$. It is
  clear that the reasoning can be repeated, yielding that either
  $T[s_2..t_2]$ is a primary occurrence or our chain can be extended
  to $T[s_1..t_1] \leftarrow T[s_2..t_2] \leftarrow T[s_3..t_3]$ for
  some substring $T[s_3..t_3]$ such that $p_3 \in [s_3,t_3]$. We
  repeat the construction for $p_4, p_5, \dots$ until either (i) we
  end up in a chain $T[i..j] \leftarrow \dots \leftarrow T[s_k..t_k]$,
  with $k<k'$, ending in a primary occurrence $T[s_k..k_k]$ of
  $T[s_1..t_1]$, or (ii) we obtain a chain $T[s_1..t_1] \leftarrow
  \dots \leftarrow T[s_{k'}..t_{k'}]$ such that $p_{k'}\in [s_{k'},
    t_{k'}]$ (i.e., we consume all the $p_1,\dots, p_{k'}$). In case
  (ii), note that $T[p_{k'}] \leftarrow c$ is one of the macro
  scheme's directives, therefore $T[s_{k'}..t_{k'}]$ is a primary
  occurrence of $T[s_1..t_1]$.
\end{proof}

The above theorem implies that LZ77 induces a string attractor of size
at most $2z$. We can achieve a better bound by exploiting the
so-called \emph{primary occurrence} property of LZ77:

\begin{lemma}\label{th: attractor LZ77}
  Let $z$ be the number of factors of the Lempel-Ziv factorization of
  a string $T$. Then, $T$ has an attractor of size $z$.
\end{lemma}
\begin{proof}
  We insert in $\Gamma_{LZ77}$ all positions at the end of a
  phrase. It is easy to see~\cite{kreft2013compressing} that every
  text substring has an occurrence crossing a phrase boundary (these
  occurrences are called \emph{primary}), therefore we obtain that
  $\Gamma_{LZ77}$ is a valid attractor for $T$.
\end{proof}

Kosaraju and Manzini~\cite{kosaraju2000compression} showed that LZ77
is \emph{coarsely optimal}, i.e., its compression ratio differs from
the $k$-th order empirical entropy by a quantity tending to zero as
the text length increases. From Lemma~\ref{th: attractor LZ77}, we can
therefore give an upper bound to the size of the smallest attractor in
terms of $k$-th order empirical entropy.\footnote{Note
  that~\cite{kosaraju2000compression} assumes a version of LZ77 that
  allows phrases to overlap their sources. It is easy to check that
  Lemma~\ref{th: attractor LZ77} holds also for this variant.}

\begin{corollary}\label{th: attractor Hk}
  Let $\gamma^*$ be the size of the smallest attractor for a string
  $T\in\Sigma^n$, and $H_k$ denote the $k$-th order empirical entropy
  of $T$. Then, $\gamma^*\log n \leq nH_k + o(n\log\sigma)$ for $k\in
  o(\log_\sigma n)$.
\end{corollary}

The run-length Burrows-Wheeler transform seems a completely different
paradigm for compressing repetitive strings: while grammars and macro
schemes explicitly copy portions of the text to other locations, with
the RLBWT we build a string permutation by concatenating characters
preceding lexicographically-sorted suffixes, and then run-length
compress it.  This strategy is motivated by the fact that equal
substrings are often preceded by the same character, therefore the BWT
contains long runs of the same letter if the string is
repetitive~\cite{makinen2010storage}.  We obtain:

\begin{theorem}\label{th: attractor RLBWT}
  Let $r$ be the number of equal-letter runs in the Burrows-Wheeler
  transform of $T$. Then, $T$ has an attractor of size $r$.
\end{theorem}
\begin{proof}
  Let $n=|T|$. Denote the BWT of $T$ by $L$ and consider the process
  of inverting the BWT to obtain $T$. The inversion algorithm is based
  on the observation that $T[n - k] = L[LF^{k}[p_0]]$ for
  $k\in[0,n-1]$, where $p_0$ is the position of $T[n]$ in $L$.  From
  the formula for $LF$ it is easy to see that if two positions $i, j$
  belong to the same equal-letter run in $L$ then $LF[j]=LF[i]+(j-i)$.
  Let $\Gamma_{BWT}=\{n-k \mid LF^{k}[p_0]=1\text{ or
  }L[LF^{k}[p_0]-1]\neq L[LF^{k}[p_0]]\}$, i.e., $\Gamma_{BWT}$ is the
  set of positions $i$ in $T$ such that if the symbol in $L$
  corresponding to $T[i]$ is $L[j]$ then $j$ is the beginning of run
  in $L$ (alternatively, we can define $\Gamma_{BWT}$ as the set of
  positions at the end of BWT runs).
	
  To show that $\Gamma_{BWT}$ is an attractor of $T$, consider any
  substring $T[i..j]$ of $T$. We show that there exists an occurrence
  of $T[i..j]$ in $T$ that contains at least one position from
  $\Gamma_{BWT}$.  Let $p=LF^{n-j}[p_0]$, i.e., $L[p]$ is the symbol
  in $L$ corresponding to $T[j]$.  Denote $\ell=j-i+1$ and let
  $[i_0,j_0]$, $[i_1,j_1]$, \ldots, $[i_{\ell-1},j_{\ell-1}]$ be the
  sequence of runs visited when the BWT inversion algorithm computes,
  from right to left, $T[i..j]$, i.e., $L[i_t..j_t]$ is the BWT-run
  containing $L[LF^{t}[p]]$, $t\in\{0,\ldots,\ell-1\}$.  Let $b={\rm
    argmin}\{LF^{t}[p]-i_t\mid t\in[0,\ell-1]\}$.  Further, let
  $\Delta=LF^{b}[p]-i_b$ and let $p'=p-\Delta$.  By definition of $b$
  and from the above property of $LF$ for two positions inside the
  same run we have that $L[LF^{t}[p']]=L[LF^{t}[p]]$ for
  $t\in[0,\ell-1]$.  This implies that if we let $j'$ be such that
  $p'=LF^{n-j'}[p_0]$ (i.e., $j'$ is such that $T[j']$ corresponds to
  $L[p']$) then $T[i..j]=T[i'..j']$ for $i':=j'-\ell+1$.  However,
  since by definition of $b$, $L[LF^{b}[p']]$ (corresponding to
  $T[j'-b]$) is at the beginning of run in $L$, $T[i'..j']$ contains a
  position from $\Gamma_{BWT}$.
\end{proof}

Finally, an analogous theorem holds for automata recognizing the
string's suffixes:

\begin{theorem}\label{th: attractor automaton}
  Let $e$ be the number of edges of a compact automaton $\mathcal A$
  recognizing all (and only the) suffixes of a string $T$. Then, $T$
  has an attractor of size $e$.
\end{theorem}
\begin{proof}
  We call \emph{root} the starting state of $\mathcal A$.  Start with
  empty attractor $\Gamma_{\mathcal A}$. For every edge $(u,v)$ of
  $\mathcal A$, do the following. Let $T[i..j]$ be any occurrence of
  the substring read from the root of $\mathcal A$ to the first
  character in the label of $(u,v)$. We insert $j$ in
  $\Gamma_{\mathcal A}$.
	
  To see that $\Gamma_{\mathcal A}$ is a valid string attractor of
  size $e$, consider any substring $T[i..j]$. By definition of
  $\mathcal A$, $T[i..j]$ defines a path from the root to some node
  $u$, plus a prefix of the label (possibly, all the characters of the
  label) of an edge $(u,v)$ originating from $u$. Let $T[i..k]$,
  $k\leq j$, be the string read from the root to $u$, plus the first
  character in the label of $(u,v)$. Then, by definition of
  $\Gamma_{\mathcal A}$ there is an occurrence $T[i'..k'] = T[i..k]$
  such that $k'\in \Gamma_{\mathcal A}$. Since the remaining (possibly
  empty) suffix $T[k+1..j]$ of $T[i..j]$ ends in the middle of an
  edge, every occurrence of $T[i..k]$ is followed by $T[k+1..j]$,
  i.e., $T[i'..i'+(j-i)]$ is an occurrence of $T[i..j]$ crossing the
  attractor's element $k'$.
\end{proof}

\subsection{Reductions to Dictionary Compressors}
\label{sec:attractors->compressors}

In this section we show reductions from string attractors to
dictionary compressors.  Combined with the results of the previous
section, this will imply that dictionary compressors can be
interpreted as approximation algorithms for the smallest string
attractor.  The next property follows easily from Definition \ref{def:
  string attractor} and will be used in the proofs of the following
theorems.

\begin{lemma}\label{lemma: superset}
  Any superset of a string attractor is also a string attractor.
\end{lemma}

We now show that we can derive a bidirectional parse from a string
attractor.

\begin{theorem}\label{th: attractor -> macro}
  Given a string $T\in\Sigma^n$ and a string attractor $\Gamma$ of
  size $\gamma$ for $T$, we can build a bidirectional parse (and
  therefore a macro scheme) for $T$ of size $\bigO(\gamma \log
  (n/\gamma))$.
\end{theorem}
\begin{proof}
  We add $\gamma$ equally-spaced attractor's elements following Lemma
  \ref{lemma: superset}.  We define phrases of the parse around
  attractor's elements in a ``concentric exponential fashion'', as
  follows.  Characters on attractor's positions are explicitly stored.
  Let $i_1 < i_2$ be two consecutive attractor's elements. Let $m =
  \lfloor(i_1 + i_2)/2\rfloor$ be the middle position between them. We
  create a phrase of length 1 in position $i_1+1$, followed by a
  phrase of length 2, followed by a phrase of length 4, and so on. The
  last phrase is truncated at position $m$. We do the same (but
  right-to-left) for position $i_2$ except the last phrase is
  truncated at position $m+1$.  For the phrases' sources, we use any
  of their occurrences crossing an attractor's element (possible by
  definition of $\Gamma$).
	
  Suppose we are to extract $T[i]$, and $i$ is inside a phrase of
  length $\leq 2^e$, for some $e$. Let $i'$ be the position from where
  $T[i]$ is copied according to our bidirectional parse. By the way we
  defined the scheme, it is not hard to see that $i'$ is either an
  explicitly stored character or lies inside a phrase of
  length\footnote{To see this, note that $2^e = 1 + 2^0 + 2^1 + 2^2 +
    \cdots + 2^{e-1}$: these are the lengths of phrases following (and
    preceding) attractor's elements (included). In the worst case,
    position $i'$ falls inside the longest such phrase (of length
    $2^{e-1}$).} $\leq 2^{e-1}$. Repeating the reasoning, we will
  ultimately ``fall'' on an explicitly stored character.  Since
  attractor's elements are at a distance of at most $n/\gamma$ from
  each other, both the parse height and the number of phrases we
  introduce per attractor's element are $\bigO(\log(n/\gamma))$.
\end{proof}

The particular recursive structure of the macro scheme of Theorem
\ref{th: attractor -> macro} can be exploited to induce a collage
system of the same size. We state this result in the following
theorem.

\begin{theorem}\label{th: attractor -> collage}
  Given a string $T\in\Sigma^n$ and a string attractor $\Gamma$ of
  size $\gamma$ for $T$, we can build a collage system for $T$ of size
  $\bigO(\gamma \log (n/\gamma))$.
\end{theorem}
\begin{proof}
  We first build the bidirectional parse of Theorem \ref{th: attractor
    -> macro}, with $\bigO(\gamma \log(n/\gamma))$ phrases of length
  at most $n/\gamma$ each. We maintain the following invariant: every
  maximal substring $T[i..j]$ covered by processed phrases is
  collapsed into a single nonterminal $Y$.

  We will process phrases in order of increasing length.  The idea is
  to map a phrase on its source and copy the collage system of the
  source introducing only a constant number of new nonterminals. By
  the bidirectional parse's definition, the source of any phrase
  $T[i..j]$ overlaps only phrases shorter than $j-i+1$
  characters. Since we process phrases in order of increasing length,
  phrases overlapping the source have already been processed and
  therefore $T[i..j]$ is a substring of the expansion of some existing
  nonterminal $K$.

  We start by parsing each maximal substring $T[i..j]$ containing only
  phrases of length 1 into arbitrary blocks of length 2 or 3. We
  create a constant number of new nonterminals per block (one for
  blocks of length two, and two for blocks of length three). Note
  that, by the way the parse is defined, this is always possible
  (since $j-i+1 \geq 2$ always holds).  We repeat this process
  recursively --- grouping nonterminals at level $k\geq0$ to form new
  nonterminals at level $k+1$ --- until $T[i..j]$ is collapsed into a
  single nonterminal.  Our invariant now holds for the base case,
  i.e., for phrases of length $t=1$: each maximal substring containing
  only phrases of length $\leq t$ is collapsed into a single
  nonterminal.

  We now proceed with phrases of length $\geq 2$, in order of
  increasing length. Let $T[i..j]$ be a phrase to be processed, with
  source at $T[i'..j']$. By definition of the parse, $T[i'..j']$
  overlaps only phrases of length at most $j-i$ and, by inductive
  hypothesis, these phrases have already been processed. It follows
  that $T[i'..j']$ is equal to a substring $K[i''..j'']$ of the
  expansion of some existing nonterminal $K$. At this point, it is
  sufficient to add a new rule $W \rightarrow K[i''..j'']$ generating
  our phrase $T[i..j]$. Since we process phrases in order of
  increasing length, $W$ is either followed ($WX_1$), preceded
  ($X_1W$), or in the middle ($X_1WX_2$) of one or two nonterminals
  $X_1, X_2$ expanding to a maximal substring containing adjacent
  processed phrases. We introduce at most two new rules of the form $Y
  \rightarrow AB$ to merge these nonterminals into a single
  nonterminal, so that our invariant is still valid. Since we
  introduce a constant number of new nonterminals per phrase, the
  resulting collage system has $\bigO(\gamma \log(n/\gamma))$ rules.
\end{proof}

A similar proof can be used to derive a (larger) straight-line
program.

\begin{theorem}\label{th: attractor -> SLP}
  Given a string $T\in\Sigma^n$ and a string attractor $\Gamma$ of
  size $\gamma$ for $T$, we can build an SLP for $T$ of size
  $\bigO(\gamma \log^2 (n/\gamma))$.
\end{theorem}
\begin{proof}
  This proof follows that for Theorem \ref{th: attractor -> collage}
  (but is slightly more complicated since we cannot use rules of the
  form $W\rightarrow K[l..r]$).  We will first show a simpler
  construction achieving an SLP of size $\bigO(\gamma \log(n/\gamma)
  \log(n))$ and then show how to refine it to achieve size
  $\bigO(\gamma \log^2(n/\gamma))$.

  For the purpose of the proof we modify the definition of SLP to
  allow also for the rules of the form $A \rightarrow XYZ$, $A
  \rightarrow ab$, and $A \rightarrow abc$ where $\{X,Y,Z\}$ are
  nonterminals and $\{a,b,c\}$ are terminals. It is easy to see that,
  if needed, the final SLP can be turned into a standard SLP without
  asymptotically affecting its size and height.

  For any nonterminal, by \emph{level} we mean the height of its
  parse-tree. This is motivated by the fact that at any point during
  the construction, any nonterminal at level $k$ will have all its
  children at level exactly $k - 1$. Furthermore, once a nonterminal
  is created, it is never changed or deleted. Levels of nonterminals,
  in particular, will thus not change during construction.

  We start by building the bidirectional parse of Theorem \ref{th:
    attractor -> macro}, with $\bigO(\gamma \log(n/\gamma))$ phrases
  of length at most $n/\gamma$ each.  We will process phrases in order
  of increasing length.  The main idea is to map a phrase on its
  source and copy the source's parse into nonterminals, introducing
  new nonterminals at the borders if needed. By the bidirectional
  parse's definition, the source of any phrase $T[i..j]$ overlaps only
  phrases shorter than $j-i+1$ characters. Since we process phrases in
  order of increasing length, phrases overlapping the source have
  already been processed and therefore their parse into nonterminals
  is well-defined. We will maintain the following invariant: once we
  finish the processing of a phrase with the source $T[i'..j']$, the
  phrase will be represented by a single nonterminal $Y$ (expanding to
  $T[i'..j']$).

  In the first version of our construction we will also maintain the
  following invariant: every maximal substring $T[i..j]$ covered by
  processed phrases is collapsed into a single nonterminal $X$. Hence,
  whenever the processing of some phrase is finished and its source is
  an expansion of some nonterminal $Y$ we have to merge it with at
  most two adjacent nonterminals representing contiguous processed
  phrases to keep our invariant true. It is clear that, once all
  phrases have been processed, the entire string is collapsed into a
  single nonterminal $S$. We now show how to process a phrase and
  analyze the number of introduced nonterminals.

  We start by parsing each maximal substring $T[i..j]$ containing only
  phrases of length 1 into arbitrary blocks of length 2 or 3. We
  create a new nonterminal for every block. We then repeat this
  process recursively --- grouping nonterminals at level $k\geq0$ to
  form new nonterminals at level $k+1$ --- until $T[i..j]$ is
  collapsed into a single nonterminal.  Our invariant now holds for
  the base case $t=1$: each maximal substring containing only phrases
  of length $\leq t$ is collapsed into a single nonterminal.

  We now proceed with phrases of length $\geq 2$, in order of
  increasing length. Let $T[i..j]$ be a phrase to be processed, with
  source at $T[i'..j']$. By definition of the parse, $T[i'..j']$
  overlaps only phrases of length at most $j-i$ and, by inductive
  hypothesis, these phrases have already been processed. We group
  characters of $T[i..j]$ in blocks of length 2 or 3 copying the parse
  of $T[i'..j']$ at level 0. Note that this might not be possible for
  the borders of length 1 or 2 of $T[i..j]$: this is the case if the
  block containing $T[i']$ starts before position $i'$ (symmetric for
  $T[j']$). In this case, we create $\bigO(1)$ new nonterminals as
  follows. If $T[i'-1,i',i'+1]$ form a block, then we group $T[i,i+1]$
  in a block of length 2 and collapse it into a new nonterminal at
  level 1. If, on the other hand, $T[i'-1,i']$ form a block, we
  consider two sub-cases. If $T[i'+1,i'+2]$ form a block, then we
  create the block to $T[i,i+1,i+2]$ and collapse it into a new
  nonterminal at level 1. If $T[i'+1,i'+2,i'+3]$ form a block, then we
  create the two blocks $T[i,i+1]$ and $T[i+2,i+3]$ and collapse them
  into 2 new nonterminals at level 1. Finally, the case when
  $T[i'-2,i'-1,i']$ form a block is handled identically to the
  previous case. We repeat this process for the nonterminals at level
  $k\geq 1$ that were copied from $T[i'..j']$, grouping them in blocks
  of length 2 or 3 according to the source and creating $\bigO(1)$ new
  nonterminals at level $k+1$ to cover the borders. After
  $\bigO(\log(n/\gamma))$ levels, $T[i..j]$ is collapsed into a single
  nonterminal. Since we create $\bigO(1)$ new nonterminals per level,
  overall we introduce $\bigO(\log(n/\gamma))$ new nonterminals.

  At this point, let $Y$ be the nonterminal just created that expands
  to $T[i..j]$. Since we process phrases in order of increasing
  length, $Y$ is either followed ($YX$), preceded ($XY$), or in the
  middle ($X_1YX_2$) of one or two nonterminals expanding to a maximal
  substring containing contiguous processed phrases. We now show how
  to collapse these two or three nonterminals in order to maintain our
  invariant, at the same time satisfying the property that
  nonterminals at level $k$ expand to two or three nonterminals at
  level $k-1$. We show the procedure in the case where $Y$ is preceded
  by a nonterminal $X$, i.e., we want to collapse $XY$ into a single
  nonterminal. The other two cases can then easily be derived using
  the same technique. Let $k_X$ and $k_Y$ be the levels of $X$ and
  $Y$. If $k_X=k_Y$, then we just create a new nonterminal
  $W\rightarrow XY$ and we are done. Assume first that $k_X \leq
  k_Y$. Let $Y_1\dots Y_t$, with $t\geq2$, be the sequence of
  nonterminals that are the expansion of $Y$ at level $k_X$. Our goal
  is to collapse the sequence $XY_1\dots Y_t$ into a single
  nonterminal. We will show that this is possible while introducing at
  most $\bigO(\log(n/\gamma))$ new nonterminals. The parsing of
  $Y_1\dots Y_t$ into blocks is already defined (by the expansion of
  $Y$), so we only need to copy it while adjusting the left border in
  order to include $X$. We distinguish two cases. If $Y_1$ and $Y_2$
  are grouped into a single block, then we replace this block with the
  new block $XY_1Y_2$ and collapse it in a new nonterminal at level
  $k_X+1$. If, on the other hand, $Y_1$, $Y_2$, and $Y_3$ are grouped
  into a single block then we replace it with the two blocks $XY_1$
  and $Y_2Y_3$ and collapse them in two new nonterminals at level
  $k_X+1$. We repeat the same procedure at levels $k_X+1, k_X+2,
  \dots, k_Y$, until everything is collapsed in a single
  nonterminal. At each level we introduce one or two new nonterminals,
  therefore overall we introduce at most $2(k_Y-k_X)+1 \in
  \bigO(\log(n/\gamma))$ new nonterminals. The case $k_X > k_Y$ is
  solved analogously except there is no upper bound of $n/\gamma$ on
  the length of the expansion of $X$ and hence in the worst case the
  procedure introduces $2(k_X-k_Y)+1 \in \bigO(\log n)$ new
  nonterminals. Overall, however, this procedure generates the SLP for
  $T$ of size $\bigO(\gamma \log(n/\gamma) \log(n))$.

  To address the problem above we introduce $\gamma$ special blocks of
  size $2n/\gamma$ starting at text positions that are multiples of
  $n/\gamma$, and we change the invariant ensuring that any contiguous
  sequence of already processed phrases is an expansion of some
  nonterminal, and instead require that at any point during the
  computation the invariant holds within all special blocks; more
  precisely, if for any special block we consider the smallest
  contiguous sequence $P_1 \cdots P_t$ of phrases that overlaps both
  its endpoints (the endpoints of the block, that is), then the old
  invariant applied to any contiguous subsequence of $P_1 \cdots P_t$
  of already processed phrases has to hold. This is enough to
  guarantee that during the algorithm the source of every phrase is
  always a substring of an expansion of some nonterminal, and whenever
  we merge two nonterminals $XY$ they always both each expand to a
  substring of length $\bigO(n/\gamma)$ which guarantees that the
  merging introduces $\bigO(\log(n/\gamma))$ new nonterminals.
  Furthermore, it is easy to see that once a phrase has been
  processed, in order to maintain the new invariant, we now need to
  perform at most 6 mergings of nonterminals (as opposed to at most 2
  from before the modification), since each phrase can overlap at most
  three special blocks. Finally, at the end of the construction we
  need to make sure the whole string $T$ is an expansion of some
  nonterminal. To achieve this we do $\log(\gamma \log(n/\gamma))$
  rounds of a pairwise merging of nonterminals corresponding to
  adjacent phrases (in the first round) or groups of phrases (in
  latter rounds). This adds $\bigO(\gamma \log(n/\gamma))$
  nonterminals.  The level of the nonterminal expanding to $T$ (i.e.,
  the height of the resulting SLP) is
  $\bigO(\log(\gamma\log(n/\gamma)) + \log(n/\gamma)) = \bigO(\log
  n)$.
\end{proof}

Using the above theorems, we can derive the approximation rates of
some compressors for repetitive strings with respect to the smallest
string attractor.

\begin{corollary}\label{cor: approximation rates}
  The following bounds hold between the size $g^*$ of the smallest
  SLP, the size $g_{rl}^*$ of the smallest run-length SLP, the size
  $z$ of the Lempel-Ziv parse, the size $b^*$ of the smallest macro
  scheme, the size $c^*$ of the smallest collage system, and the size
  $\gamma^*$ of the smallest string attractor:
  \begin{enumerate}
    \item $b^*, c^* \in \bigO(\gamma^* \log (n/\gamma^*))$,
    \item $g^*, g^*_{rl}, z \in \bigO(\gamma^* \log^2 (n/\gamma^*))$.
  \end{enumerate}
\end{corollary}
\begin{proof}
  For the first bounds, build the bidirectional parse of Theorem
  \ref{th: attractor -> macro} and the collage system of Theorem
  \ref{th: attractor -> collage} using a string attractor of minimum
  size $\gamma^*$. For the second bound, use the same attractor to
  build the SLP of Theorem \ref{th: attractor -> SLP} and exploit the
  well-known relation $z\leq g^*$~\cite{rytter2003application}.
\end{proof}

Our reductions and the above corollary imply our first approximation
algorithms for the smallest string attractor. Note that only one of
our approximations is computable in polynomial time (unless P=NP): the
attractor induced by the LZ77 parsing. In the next section we show how
to obtain asymptotically better approximations in polynomial time.

All our reductions combined imply the following relations between
repetitiveness measures:

\begin{corollary}\label{cor: bounds1}
  The following bounds hold between the size $g^*$ of the smallest
  SLP, the size $z$ of the Lempel-Ziv parse, the size $c^*$ of the
  smallest collage system, the size $b^*$ of the smallest macro
  scheme, and the number $r$ of equal-letter runs in the BWT:
  \begin{enumerate}
    \item $z, g^* \in
      \bigO(b^*\log^2\frac{n}{b^*})\ \cap\ \bigO(r\log^2\frac{n}{r})\ \cap\ \bigO(c^*\log^2\frac{n}{c^*})$,
    \item $c^* \in
      \bigO(b^*\log\frac{n}{b^*})\ \cap\ \bigO(r\log\frac{n}{r})$,
    \item $b^* \in \bigO(c^*\log\frac{n}{c^*})$.
  \end{enumerate}
\end{corollary}
\begin{proof}
  For bounds 1, build the SLP of Theorem \ref{th: attractor -> SLP} on
  string attractors of size $b^*$, $r$, and $c^*$ induced by the
  smallest macro scheme (Theorem \ref{th: attractor MS}), RLBWT
  (Theorem \ref{th: attractor RLBWT}), and smallest collage system
  (Theorem \ref{th:attr_c}). The results follow from the definition of
  the smallest SLP and the bound $z \leq
  g^*$~\cite{rytter2003application}. Similarly, bounds 2 and 3 are
  obtained with the reductions of Theorems \ref{th: attractor RLBWT},
  \ref{th: attractor -> macro}, and \ref{th: attractor -> collage}.
\end{proof}

Some of these (or even tighter) bounds have been very recently
obtained by Gagie et al. in~\cite{GNPlatin18} and in the extended
version~\cite{gagie2017optimalArxiv} of~\cite{gagie2017optimal} using
different techniques based on locally-consistent parsing. Our
reductions, one the other hand, are slightly simpler and naturally
include a broader class of dictionary compressors, e.g., all relations
concerning $c^*$ have not been previously known.

\section{Computational Complexity}

By $\textsc{Attractor}=\{\langle T,p\rangle : \text{String }T\text{
  has an attractor of size}\leq p\}$ we denote the language
corresponding to the decision version of the smallest attractor
problem.  To prove the NP-completeness of \textsc{Attractor} we first
generalize the notion of string attractor.

\begin{definition}
  We say that a set $\Gamma \subseteq [1..n]$ is a
  \emph{$k$-attractor} of a string $T\in\Sigma^n$ if every substring
  $T[i..j]$ such that $i \leq j < i+k$ has an occurrence
  $T[i'..j']=T[i..j]$ with $j''\in[i'..j']$ for some $j''\in
  \Gamma$.\footnote{We permit non-constant $k=f(n)$ where $n=|T|$ as
    long as $\lim_{n\to\infty} f(n)=\infty$ and $f(n)$ is
    non-decreasing.}
\end{definition}

By \textsc{Minimum-$k$-Attractor} we denote the optimization problem
of finding the smallest $k$-attractor of a given input string.  By
$$k\textsc{-Attractor} = \{\langle T,p\rangle :\ T\text{ has a
}k\text{-attractor of size}\leq p\}$$ we denote the corresponding
decision problem. Observe that \textsc{Attractor} is a special case of
\textsc{$k$-Attractor} where $k=n$. The NP-completeness of
\textsc{$k$-Attractor} for any $k\geq 3$ (this includes any constant
$k\geq 3$ as well as any non-constant $k$) is obtained by a reduction
from the \textsc{$k$-SetCover} problem that is NP-complete~\cite{df97}
for any constant $k \geq 3$: given integer $p$ and a collection
$C=\{C_1, C_2, \ldots, C_m\}$ of $m$ subsets of a universe set
$\mathcal{U}=\{1, 2, \ldots, u\}$ such that
$\bigcup_{i=1}^{m}C_i=\mathcal{U}$, and for any $i\in\{1,\ldots,m\}$,
$|C_i| \leq k$, return ``yes'' iff there exists a subcollection
$C'\subseteq C$ such that $\bigcup C' = \mathcal{U}$ and $|C'|\leq p$.

We obtain our reduction as follows. For any constant $k \geq 3$, given
an instance $\langle\mathcal{U},C\rangle$ of \textsc{$k$-SetCover} we
build a string $T_C$ of length $\bigO(uk^2+tk+t')$ where
$t=\sum_{i=1}^m n_i$ and $t'=\sum_{i=1}^{m}n_i^2$ with the following
property: $\langle \mathcal{U},C \rangle$ has a cover of size $\leq p$
if and only if $T_C$ has a $k$-attractor of size $\leq
4u(k-1)+p+6t-3m$. This establishes the NP-completeness of
$k$-\textsc{Attractor} for any constant $k\geq 3$.  We then show that
for $T_C$ the size of the smallest $k$-attractor is equal to that of
the smallest $k'$-attractor for every $k\leq k'\leq |T_C|$, which
allows us to prove the NP-completeness for non-constant $k$.

\begin{theorem}
  \label{thm:k-attractor-npc}
  For $k \geq 3$, \textsc{$k$-Attractor} is NP-complete.
\end{theorem}
\begin{proof}
  Assume first that $k \geq 3$ is constant.  We show a polynomial time
  reduction from \textsc{$k$-SetCover} to
  \textsc{$k$-Attractor}.\footnote{The proof only requires that $k\geq
    3$ but we point out that the reduction is valid also for $k=2$.}
  Denote the sizes of individual sets in the collection $C$ by
  $n_i=|C_i|>0$ and let
  $C_i=\{c_{i,1},c_{i,2},\ldots,c_{i,n_i}\}$. Recall that
  $u=|\mathcal{U}|$ and $m=|C|$.
	
  Let
  \[
    \Sigma = \bigcup_{i=1}^{u}\bigcup_{j=1}^{k}\{x_i^{(j)}\} \cup
    \bigcup_{i=1}^{m}\bigcup_{j=1}^{n_i+1}\{\$_{i,j}\} \cup
    \bigcup_{i=1}^{u}\bigcup_{j=2}^{k}\{\$_{i,j}',\$_{i,j}''\} \cup
    \bigcup_{i=1}^{m}\bigcup_{j=2}^{n_i}\{\$_{i,j}''',\$_{i,j}^{(4)}\}
    \cup \{\#\}
  \]
  be our alphabet. Note that in the construction below, $x_i^{(j)}$ or
  $\$_{i,j}^{(4)}$ denotes a single symbol, while $\#^{k-1}$ denotes a
  concatenation of $k-1$ occurrences of symbol $\#$. We will now build
  a string $T_C$ over the alphabet $\Sigma$.
	
  Let
  \[
    T_C = \prod_{i=1}^{u}P_i \cdot \prod_{i=1}^{m} R_i S_i,
  \]
  where $\cdot$/$\prod$ denotes the concatenation of strings and
  $P_i$, $R_i$, and $S_i$ are defined below.
	
  Intuitively, we associate each $t\in \mathcal{U}$ with the substring
  $x_{t}^{(1)}\cdots x_{t}^{(k)}$ and each collection $C_i$ with
  $S_i$.  Each $S_i$ will contain all $n_i$ strings corresponding to
  elements of $C_i$ as substrings. The aim of $S_i$ is to simulate ---
  via how many positions are used within $S_i$ in the solution to the
  \textsc{$k$-Attractor} on $T_C$ --- the choice between not including
  $C_i$ in the solution to \textsc{$k$-SetCover} on $C$ (in which case
  $S_i$ is covered using a minimum possible number of positions that
  necessarily leaves uncovered all substrings corresponding to
  elements of $C_i$) or including $C_i$ (in which case, by using only
  one additional position in the cover of $S_i$, the solution covers
  all substrings unique to $S_i$ \emph{and simultaneously} all $n_i$
  substrings of $S_i$ corresponding to elements of $C_i$).  Gadgets
  $R_i$ and $P_i$ are used to cover ``for free'' certain substrings
  occurring in $S_i$ so that any algorithm solving
  \textsc{$k$-Attractor} for $T_C$ will not have to optimize for their
  coverage within $S_i$. This will be achieved as follows: each gadget
  $P_i$ (similar for $R_i$) will have $x_{P_i}$ non-overlapping
  substrings (for some $x_{P_i}$) that appear only in $P_i$ and
  nowhere else in $T_C$. This will imply that any $k$-attractor for
  $T_C$ has to include at least $x_{P_i}$ positions within $P_i$. On
  the other hand, we will show that there exists an optimal choice of
  $x_{P_i}$ positions within $P_i$ that covers all those unique
  substrings, plus the substrings of $P_i$ occurring also $S_i$ that
  we want to cover ``for free'' within $P_i$.
	
  For $i\in\{1,2,\ldots,m\}$, let (brackets added for clarity)
  \[
    S_i = \left( \prod_{j=1}^{n_i} \#^{k-1} \$_{i,1} \cdots \$_{i,j}
    x_{c_{i,j}}^{(1)} \cdots x_{c_{i,j}}^{(k)} \$_{i,j} \right)
    \#^{k-1} \$_{i,1} \cdots \$_{i,n_i+1}.
  \]
	
  An example of $S_i$ for $k=6$ and $n_i=4$ is reported below. The
  meaning of overlined and underlined characters is explained next.
    \begin{align*}
      S_i =\ & \# \# \# \# \# \underline{\$_{i,1}}
	\overline{x_{c_{i,1}}^{(1)}} x_{c_{i,1}}^{(2)} x_{c_{i,1}}^{(3)}
	x_{c_{i,1}}^{(4)} x_{c_{i,1}}^{(5)} x_{c_{i,1}}^{(6)}
	\overline{\underline{\$_{i,1}}} \\
	& \# \# \# \# \# \$_{i,1}
	\underline{\$_{i,2}} \overline{x_{c_{i,2}}^{(1)}}
	x_{c_{i,2}}^{(2)} x_{c_{i,2}}^{(3)} x_{c_{i,2}}^{(4)}
	x_{c_{i,2}}^{(5)} x_{c_{i,2}}^{(6)}\overline{\underline{\$_{i,2}}}\\
	& \# \# \# \# \# \$_{i,1}
	\$_{i,2} \underline{\$_{i,3}} \overline{x_{c_{i,3}}^{(1)}}
	x_{c_{i,3}}^{(2)} x_{c_{i,3}}^{(3)} x_{c_{i,3}}^{(4)}
	x_{c_{i,3}}^{(5)} x_{c_{i,3}}^{(6)}
	\overline{\underline{\$_{i,3}}}\\
	& \# \# \# \# \# \$_{i,1} \$_{i,2}
	\$_{i,3} \underline{\$_{i,4}} \overline{x_{c_{i,4}}^{(1)}}
	x_{c_{i,4}}^{(2)} x_{c_{i,4}}^{(3)}x_{c_{i,4}}^{(4)}
	x_{c_{i,4}}^{(5)} x_{c_{i,4}}^{(6)}
	\overline{\underline{\$_{i,4}}}\\[1.3ex]
  & \# \# \# \# \# \overline{\$_{i,1}}
	\$_{i,2} \$_{i,3} \$_{i,4} \overline{\underline{\$_{i,5}}}
    \end{align*}
	
  Any $k$-attractor of $T_C$ contains at least $2n_i+1$ positions
  within $S_i$ because: (i) $S_i$ contains $2n_i$ non-overlapping
  substrings of length $k$, each of which
  \emph{necessarily}\footnote{That is, independent of what is $C_i$.
    Importantly, although any of the substrings in
    $\bigcup_{j=1}^{n_i}\{x_{c_{i,j}}^{(1)}\cdots x_{c_{i,j}}^{(k)}\}$
    could have the only occurrence in $S_i$, they are not
    \emph{necessarily} unique to $S_i$. This situation is analogous to
    \textsc{$k$-SetCover} when some $t\in\mathcal{U}$ is covered by
    only one set in $C$, and thus that set has to be included in the
    solution.}  occurs in $S_i$ only once and nowhere
  else\footnote{This can be verified by consulting the definitions of
    families $\{R_t\}_{t=1}^{m}$ and $\{P_t\}_{t=1}^{u}$ that follow.}
  in $T_C$: $$\bigcup_{j=1}^{n_i}\{\$_{i,j}\#^{k-1}\}
  \cup\bigcup_{j=1}^{n_i}\{\$_{i,j}x_{c_{i,j}}^{(1)}\cdots
  x_{c_{i,j}}^{(k-1)}\},$$ and (ii) $S_i$ contains symbol
  $\$_{i,n_i+1}$, which occurs only once in $S_i$ and nowhere else in
  $T_C$, and does not \mbox{overlap any} of the $2n_i$ substrings mentioned
  before. With this in mind we now observe that $S_i$ has the
  following \mbox{two properties}:
  \begin{enumerate}
    \item There exists a ``minimum'' set $\Gamma_{S,i}$ of $2n_i+1$
      positions within the occurrence of $S_i$ in $T_C$ that covers
      all substrings of $S_i$ of length $\leq k$ that necessarily
      occur only in $S_i$ and nowhere else in $T_C$. The \mbox{set
      $\Gamma_{S,i}$ includes}: the leftmost occurrence of $\$_{i,j}$
      for $j\in\{1,\ldots,n_i+1\}$ and the second occurrence from the
      left of $\$_{i,j}$ for $j\in\{1,\ldots,n_i\}$ ($\Gamma_{S,i}$ is
      shown in the above example using underlined
      positions). Furthermore, $\Gamma_{S,i}$ is \emph{the only} such
      set.  This is because in any such set there needs to be at least
      one position inside each of the $2n_i+1$ non-overlapping
      substrings of $S_i$ mentioned above. In the first $n_i$ of those
      substrings, $\bigcup_{j=1}^{n_i}\{\$_{i,j}\#^{k-1}\}$, the first
      position intersects $k$ distinct substrings of length $k$ that
      necessarily occur only once in $S_i$ and nowhere else in $T_C$,
      and hence in those substrings the position in the attractor is
      fixed. Next, the position in any such set is also trivially
      fixed for the only occurrence of $\$_{i,n_i+1}$ in $T_C$.  Let
      us then finally look at each of the remaining $n_i$ substrings,
      $\bigcup_{j=1}^{n_i}\{\$_{i,j}x_{c_{i,j}}^{(1)}\cdots
      x_{c_{i,j}}^{(k-1)}\}$, starting from the rightmost
      ($j=n_i$). In the substring $\$_{i,n_i}x_{c_{i,n_i}}^{(1)}\cdots
      x_{c_{i,n_i}}^{(k-1)}$ the first position intersects only $k-1$
      substrings of $S_i$ of length $k$ that necessarily occur only
      once in $S_i$ and nowhere else in $T_C$. However, all other
      occurrences (for $j=n_i$ there is only one; in general there is
      $n_i-j+1$ occurrences) in $T_C$ of the remaining non-unique
      substring intersecting the first position,
      $\#^{k-n_i}\$_{i,1}\cdots \$_{i,n_i}$, are in $S_i$ (and nowhere
      else in $T_C$), to the right of the discussed occurrence of
      $\$_{i,n_i}x_{c_{i,n_i}}^{(1)}\cdots x_{c_{i,n_i}}^{(k-1)}$, and
      are not covered.  Thus, the attractor needs to include the first
      position in this substring. Repeating this argument for
      $j=n_i-1, \ldots, 1$ yields the claim.  Now we observe that the
      only substrings of $S_i$ of length $\leq k$ not covered by
      $\Gamma_{S,i}$ are strings $\{\#^{k-1}\} \cup
      \bigcup_{j=1}^{n_i}\{x_{c_{i,j}}^{(1)}\cdots
      x_{c_{i,j}}^{(k)}\}$ and all their proper substrings.  We have
      thus demonstrated that if in any $k$-attractor of $T_C$, $S_i$
      is covered using the minimum number of $2n_i+1$ positions, these
      positions \emph{must} be precisely $\Gamma_{S,i}$ and hence, in
      particular, any of the strings in the set
      $\bigcup_{j=1}^{n_i}\{x_{c_{i,j}}^{(1)}\cdots
      x_{c_{i,j}}^{(k)}\}$ is then not covered within $S_i$.
    \item There exists a ``nearly-universal'' set $\Gamma_{S,i}'$ of
      $2n_i+2$ positions within the occurrence of $S_i$ in $T_C$ that
      covers: (i) all substrings of $S_i$ of length $\leq k$ that
      necessarily occur only in $S_i$ and nowhere else in $T_C$, and
      (ii) $\bigcup_{j=1}^{n_i}\{x_{c_{i,j}}^{(1)}\cdots
      x_{c_{i,j}}^{(k)}\}$. The set $\Gamma_{S,i}'$ includes: the only
      occurrence of $x_{c_{i,j}}^{(1)}$ for $j\in\{1,\ldots,n_i\}$,
      the second occurrence of $\$_{i,j}$ for $j\in\{1,\ldots,n_i\}$,
      the only occurrence of $\$_{i,n_i+1}$, and the last occurrence
      of $\$_{i,1}$ ($\Gamma_{S,i}'$ is shown in the above example
      using overlined positions). The only substrings of $S_i$ of
      length $\leq k$ not covered by $\Gamma_{S,i}'$ are strings
      $\{\#^{k-1}\}\cup \bigcup_{j=1}^{n_i}\{x_{c_{i,j}}^{(2)} \cdots
      x_{c_{i,j}}^{(k)}\}$ and all their proper substrings, and all
      substrings of length $> 1$ of the string
      $\$_{i,2}\cdots\$_{i,n_i}$. For these strings we introduce
      separate gadget strings described next.
  \end{enumerate}
	
  To finish the construction, we will ensure that for any
  $i\in\{1,\ldots,m\}$, certain substrings of $S_i$ are covered ``for
  free'' elsewhere in $T_C$. To this end we introduce families
  $\{P_i\}_{i=1}^{u}$ and $\{R_i\}_{i=1}^{m}$. Specifically, all
  strings (and all their proper substrings) in the set
  $\{\#^{k-1}\}\cup\bigcup_{i=1}^{u}\{x_{i}^{(2)}\cdots x_{i}^{(k)}\}$
  will be covered for free in $\{P_i\}_{i=1}^{u}$. Analogously, all
  strings (and all their proper substrings) in
  $\bigcup_{i=1}^{m}\{\$_{i,2}\cdots\$_{i,n_i}\}$ will be covered for
  free in $\{R_i\}_{i=1}^{m}$.  Assuming these substrings are covered:
  (i) if we use $\Gamma_{S,i}$ to cover unique substrings of $S_i$,
  the only substrings of $S_i$ of length $\leq k$ not covered by
  $\Gamma_{S,i}$ will be $\bigcup_{j=1}^{n_i}\{x_{c_{i,j}}^{(1)}\cdots
  x_{c_{i,j}}^{(k)}\}$, and (ii) if we use $\Gamma_{S,i}'$, all
  substrings of $S_i$ of length $\leq k$ will be covered.
	
  We now show the existence of the families $\{P_i\}_{i=1}^{u}$ and
  $\{R_i\}_{i=1}^{m}$.  For $i\in\{1,2,\ldots,u\}$, let
  \[
    P_i = \prod_{j=2}^{k} \#^{k-1} \$_{i,j}' x_{i}^{(2)} \cdots
    x_{i}^{(j)} \$_{i,j}'' \#^{k-1} \$_{i,j}' x_{i}^{(2)} \cdots
    x_{i}^{(j)} \$_{i,j}' \$_{i,j}' \$_{i,j}''.
  \]
	
  A prefix of $P_i$ for $k=6$ is
  \begin{align*}
    P_i =\ 
    & \# \# \# \# \# \$_{i,2}' \underline{x_i^{(2)}} \$_{i,2}''
      \# \# \# \# \underline{\#} \$_{i,2}' x_i^{(2)} \underline{\$_{i,2}'} \$_{i,2}' \underline{\$_{i,2}''}\\[0.5ex]
    & \# \# \# \# \# \$_{i,3}' x_i^{(2)} \underline{x_i^{(3)}} \$_{i,3}''
      \# \# \# \# \# \underline{\$_{i,3}'} x_i^{(2)} x_i^{(3)}\underline{\$_{i,3}'} \$_{i,3}' \underline{\$_{i,3}''}\\[0.5ex]
    & \# \# \# \# \# \$_{i,4}' x_i^{(2)} x_i^{(3)} \underline{x_i^{(4)}} \$_{i,4}''
      \# \# \# \# \# \underline{\$_{i,4}'} x_i^{(2)} x_i^{(3)} x_i^{(4)} \underline{\$_{i,4}'} \$_{i,4}' \underline{\$_{i,4}''}.
  \end{align*}
	
  Any $k$-attractor of $T_C$ contains at least $4(k-1)$ positions
  within $P_i$ because there are $4(k-1)$ non-overlapping substrings
  of length two of $P_i$ that occur only in $P_i$ and nowhere else in
  $T_C$.\footnote{This for example enforces $k \geq 2$ in our proof.}
  These substrings are, for $j\in\{2,\ldots,k\}$: $\# \$_{i,j}'$,
  $x_{i}^{(j)} \$_{i,j}''$, $x_{i}^{(j)} \$_{i,j}'$,
  $\$_{i,j}'\$_{i,j}''$.
	
  On the other hand, there exists a ``universal'' set $\Gamma_{P,i}$
  of $4(k-1)$ positions within the occurrence $P_i$ in $T_C$ that
  covers \emph{all} substrings of $P_i$ of length $\leq
  k$.\footnote{Note a small subtlety here. Because each of the gadgets
    $\{P_i\}_{i=1}^{u}$, $\{S_i\}_{i=1}^{m}$, $\{R_i\}_{i=1}^{m}$
    begins with $\#^{k-1}$ and each $P_i$ is followed by some other
    gadget in $T_C$, the following set of substrings of $P_i$:
    $\{\$_{i,k}''\#^{t}\}_{t=1}^{k-1}$ will indeed be covered by
    $\Gamma_{P,i}$ but for $k>2$ the covered occurrences are not
    substrings of $P_i$. An analogous property holds for
    $\{R_i\}_{i=1}^{m}$.} In particular, $\Gamma_{P,i}$ covers the
  strings $x_i^{(2)}\cdots x_{i}^{(k)}$ and $\#^{k-1}$, and all their
  proper substrings.  The set $\Gamma_{P,i}$ includes: the position of
  the leftmost occurrence of $x_{i}^{(j)}$ for $j\in\{2,\ldots,k\}$,
  the position preceding the second occurrence of $\$_{i,2}'$ from the
  left, the third occurrence of $\$_{i,2}'$ from the left, the second
  and third occurrences of $\$_{i,j}'$ from the left for
  $j\in\{3,\ldots,k\}$, and the second occurrence of $\$_{i,j}''$ from
  the left for $j\in\{2,\ldots,k\}$.  The positions in $\Gamma_{P,i}$
  are underlined in the above example.
	
  For $i\in\{1,2,\ldots,m\}$, let
  \[
    R_i = \prod_{j=2}^{n_i} \#^{k-1} \$_{i,j}''' \$_{i,2} \cdots
    \$_{i,j} \$_{i,j}^{(4)} \#^{k-1} \$_{i,j}''' \$_{i,2} \cdots
    \$_{i,j} \$_{i,j}''' \$_{i,j}''' \$_{i,j}^{(4)}.
  \]
	
  An example of $R_i$ for $k=6$ and $n_i=4$ is
  \begin{align*}
    R_i =\
    & \# \# \# \# \# \$_{i,2}''' \underline{\$_{i,2}} \$_{i,2}^{(4)}
      \# \# \# \# \underline{\#} \$_{i,2}''' \$_{i,2} \underline{\$_{i,2}'''} \$_{i,2}''' \underline{\$_{i,2}^{(4)}} \\[0.5ex]
    & \# \# \# \# \# \$_{i,3}''' \$_{i,2} \underline{\$_{i,3}} \$_{i,3}^{(4)}
      \# \# \# \# \#\underline{\$_{i,3}'''} \$_{i,2} \$_{i,3} \underline{\$_{i,3}'''} \$_{i,3}''' \underline{\$_{i,3}^{(4)}}\\[0.5ex]
    & \# \# \# \# \# \$_{i,4}''' \$_{i,2} \$_{i,3} \underline{\$_{i,4}} \$_{i,4}^{(4)}
      \# \# \# \# \# \underline{\$_{i,4}'''} \$_{i,2} \$_{i,3} \$_{i,4} \underline{\$_{i,4}'''} \$_{i,4}''' \underline{\$_{i,4}^{(4)}}.
  \end{align*}
	
  Note, that if $n_i=1$ then $R_i$ is the empty string. Suppose that
  $R_i$ is non-empty, i.e., $n_i\geq 2$.  The construction of $R_i$ is
  analogous to $P_i$.  Any $k$-attractor of $T_C$ contains at least
  $4(n_i-1)$ positions within $R_i$.  On the other hand, there exists
  a ``universal'' set $\Gamma_{R,i}$ of $4(n_i-1)$ positions of $R_i$
  that covers \emph{all} substrings of $R_i$ of length $\leq k$.  In
  particular, $\Gamma_{R,i}$ covers the string $\$_{i,2}\cdots
  \$_{i,n_i}$ and all its proper substrings.  The set $\Gamma_{R,i}$
  includes: the position of the leftmost occurrence of $\$_{i,j}$ for
  $j\in\{2,\ldots,n_i\}$, the position preceding the second occurrence
  of $\$_{i,2}'''$ from the left, the third occurrence of
  $\$_{i,2}'''$ from the left, the second and third occurrence of
  $\$_{i,j}'''$ for $j\in\{3,\ldots,n_i\}$, and the second occurrence
  of $\$_{i,j}^{(4)}$ from the left for $j\in\{2,\ldots,n_i\}$.  The
  positions in $\Gamma_{R,i}$ are underlined in the example.
	
  With the above properties, we are now ready to prove the following
  claim: an instance $\langle \mathcal{U},C \rangle$ of
  \textsc{$k$-SetCover} has a solution of size $\leq p$ if and only if
  $T_C$ has a $k$-attractor of size $\leq 4u(k-1)+p+6t-3m$, where
  $t=\sum_{i=1}^{m}n_i$.
	
  ``$(\Rightarrow)$'' Let $C' \subseteq C$ be a cover of $\mathcal{U}$
  of size $p'\leq p$ and let
  \begin{align*}
    \Gamma_{C'} = 
    \bigcup\{\Gamma_{S,i}'\mid C_i \in C'\} \cup
    \bigcup\{\Gamma_{S,i}\mid C_i\not\in C'\} \cup
    \bigcup_{i=1}^{u}\Gamma_{P,i} \cup \bigcup_{i=1}^{m}\Gamma_{R,i}.
  \end{align*}
	
  $\Gamma_{C'}$ contains universal attractors $\Gamma_{P,\cdot}$ and
  $\Gamma_{R,\cdot}$ introduced above for $\{P_i\}_{i=1}^{u}$ and
  $\{R_i\}_{i=1}^{m}$, and nearly-universal attractors
  $\Gamma_{S,\cdot}'$ for elements of $\{S_i\}_{i=1}^{m}$
  corresponding to elements of $C'$. All other strings in
  $\{S_i\}_{i=1}^{m}$ are covered using minimum attractors
  $\Gamma_{S,\cdot}$. It is easy to check that
  $|\Gamma_{C'}|=4u(k-1)+p'+6t-3m$.  From the above discussion
  $\Gamma_{C'}$ covers all substrings of $T_C$ of length $\leq k$
  inside $\{P_i\}_{i=1}^{u}$, $\{R_i\}_{i=1}^{m}$, and
  $\{S_i\}_{i=1}^{m}$. In particular, $\{x_i^{(1)}\cdots
  x_i^{(k)}\}_{i=1}^{u}$ are covered because $C'$ is a cover of
  $\mathcal{U}$.  All other substrings of $T_C$ of length $\leq k$
  span at least two gadget strings and thus are also covered since all
  sets forming $\Gamma_{C'}$ include the last position of the gadget
  string.
	
  ``$(\Leftarrow)$'' Let $\Gamma$ be a $k$-attractor of $T_C$ of size
  $\leq 4u(k-1)+p+6t-3m$.  We will show that $\mathcal{U}$ must have a
  cover of size $\leq p$ using elements from $C$.  Let $\mathcal{I}$
  be the set of indices $i\in\{1,\ldots,m\}$ for which $\Gamma$
  contains more than $2n_i+1$ positions within the occurrence of $S_i$
  in $T_C$. To bound the cardinality of $\mathcal{I}$ we first observe
  that by the above discussion, $\Gamma$ cannot have less than
  $\sum_{i=1}^{m}4(n_i-1)+\sum_{i=1}^{u}4(k-1)=4u(k-1)+4t-4m$
  positions within all occurrences of $\{P_i\}_{i=1}^{u}$ and
  $\{R_i\}_{i=1}^{m}$ in $T_C$.  Thus, there is only at most $2t+m+p$
  positions left to use within $\{S_i\}_{i=1}^{m}$.  Furthermore, each
  of $S_i$, $i\in\{1,\ldots,m\}$ requires $2n_i+1$ positions, and
  hence there cannot be more than $p$ indices where $\Gamma$ uses more
  positions than necessary. Thus, $|\mathcal{I}|\leq p$. Let
  $C'_{\Gamma}=\{C_i \in C \mid i \in \mathcal{I}\}$. We now show that
  $C'_{\Gamma}$ is a cover of $\mathcal{U}$.  Take any
  $t\in\mathcal{U}$. Since $\Gamma$ is a $k$-attractor of $T_C$, the
  string $x_t^{(1)}\cdots x_t^{(k)}$ is covered in some $S_{i_t}$ such
  that $t\in C_{i_t}$. By the above discussion for this to be possible
  $\Gamma$ must use more than $2n_{i_t}+1$ positions within
  $S_{i_t}$. Thus, $i_t\in\mathcal{I}$ and hence $C_{i_t}\in
  C_{\Gamma}'$.
	
  The above reduction proves the NP-completeness of
  \textsc{$k$-Attractor} for any constant $k\geq 3$.  We now show a
  property of $T_C$ that will allow us to obtain the NP-completeness
  for non-constant $k$.  Denote the size of the smallest $k$-attractor
  of string $X$ by $\gamma_k^*(X)$. By definition a $k'$-attractor of
  string $X$ is also a $k$-attractor of $X$ for any $k < k'$ and thus
  for any $k\in\{1,\ldots,|X|-1\}$, $\gamma_k^*(X) \leq
  \gamma_{k+1}^*(X)$. The inequality in general can be strict, e.g.,
  for $X={\tt acacaacc}$, $\gamma_2^*(X)<\gamma_3^*(X)$.  We now show
  that for $T_C$ it holds $\gamma_{k}^*(T_C) = \gamma_{k'}^*(T_C)$ for
  any $k<k'\leq |T_C|$. Assume that $p$ is the size of the smallest
  $k$-set-cover of $\mathcal{U}$ and let $C'\subseteq C$ be the
  optimal cover.  Then, $\Gamma_{C'}$ (defined as above) is the
  smallest $k$-attractor of $T_C$ and, crucially, admits a particular
  structure, namely, it is a union of universal, nearly-universal and
  minimum attractors introduced above. We will now show that
  $\Gamma_{C'}$ is a $k'$-attractor of $T_C$. Since each of the sets
  forming $\Gamma_{C'}$ covers the last position of the corresponding
  gadget string, we can focus on substrings of length $>k$ entirely
  contained inside gadget strings.  To show the claim for
  $\{S_i\}_{i=1}^{m}$ it suffices to verify that all substrings of
  $\#^{k-1}\$_{i,1}\cdots \$_{i,n_i}$ of length $>k$ are covered in
  both $\Gamma_{S,i}$ and $\Gamma_{S,i}'$.  Analogously, for
  $\{P_i\}_{i=1}^{u}$ and $\{R_i\}_{i=1}^{m}$ it suffices to verify
  the claim for the families $\{\#^{k-1}\$_{i,j}'x_{i}^{(2)}\cdots
  x_{i}^{(j-1)}\}_{j=3}^{k}$ and $\{\#^{k-1}\$_{i,j}'''\$_{i,2}\cdots
  \$_{i,j-1}\}_{j=3}^{n_i}$.  Thus, $\Gamma_{C'}$ is a $k'$-attractor
  of $T_C$.
	
  To show the NP-completeness of $k$-\textsc{Attractor} for
  non-constant $k$ (in particular for $k=n$) consider any
  non-decreasing function $k=f(n)$ such that $\lim_{n \to \infty} f(n)
  = \infty$. Let $n_0=\min\{n\geq 1 \mid f(n) \geq 3\}$. Suppose that
  we have a polynomial-time algorithm for $f(n)$-\textsc{Attractor}.
  Consider an instance $\langle \mathcal{U},C \rangle$,
  $C=\{C_i\}_{i=1}^{m}$ of 3-\textsc{SetCover}.  To decide if $\langle
  \mathcal{U},C \rangle$ has a cover of size $\leq p$, we first build
  the string $T_C$. If $|T_C|<n_0$, we run a brute-force algorithm to
  find the answer in $\bigO(2^m\,{\rm poly}(t))=\bigO(2^{n_0}\,{\rm
    poly}(n_0))=\bigO(1)$ time, where $t=\sum_{i=1}^{m}|C_i|$.
  Otherwise, the answer is given by checking the inequality
  $\gamma_{3}^{*}(T_C)=\gamma_{f(|T_C|)}^{*}(T_C) \leq 8u+p+6t-3m$
  (where $u=|\mathcal{U}|$) in polynomial time.
\end{proof}

We further demonstrate that \textsc{Minimum}-$k$-\textsc{Attractor}
can be efficiently approximated up to a constant factor when $k\geq 3$
is constant, but unless P=NP, does not have a PTAS. This is achieved
by a reduction from vertex cover on bounded-degree graphs, utilizing
the smallest $k$-set cover as an intermediate problem. Using explicit
constants derived by Berman and Karpinski~\cite{bk99} for the vertex
cover, we also obtain explicit constants for our problem (and general
$k$).

\begin{theorem}
  \label{thm:apx-complete}
  For any constant $k\geq 3$, \textsc{Minimum-$k$-Attractor} is
  APX-complete.
\end{theorem}
\begin{proof}
  Denote the size of the smallest $k$-attractor of $T$ by
  $\gamma_k^*(T)$ and let $\sigma_k(T)$ be the number of different
  substrings of $T$ of length $k$. We claim that $\gamma_k^*(T) \leq
  \sigma_k(T^2) \leq 2k\gamma_k^*(T)$ (where $T^2$ is a concatenation
  of two copies of $T$).  To show the first inequality, define
  $\Gamma$ as the set containing the beginning of the leftmost
  occurrence of every distinct substring of $T^2$ of length $k$. Such
  $\Gamma$ can be easily computed in polynomial time. We claim that
  $\Gamma$ is a $k$-attractor of $T$. Consider any substring of $T$ of
  length $k'\leq k$ and let $T[i..i+k'-1]$ be its leftmost occurrence.
  Then, $T^2[i..i+k-1]$ is the leftmost occurrence of $T^2[i..i+k-1]$
  in $T^2$, as otherwise we would have an earlier occurrence of
  $T[i..i+k'-1]$ in $T$. Thus, $i\in\Gamma$.  On the other hand, each
  position in a $k$-attractor of $T^2$ covers at most $k$ distinct
  substrings of $T^2$ of length $k$. Thus $\gamma_{k}^{*}(T^2) \geq
  \lceil \sigma_k(T^2)/k \rceil$. Combining this with $\gamma_k^*(T^2)
  \leq \gamma_k^*(T)+1$ gives the second inequality.  Thus,
  \textsc{Minimum}-$k$-\textsc{Attractor} is in APX.
	
  To show that \textsc{Minimum-$k$-Attractor} is APX-hard we
  generalize the hardness argument of Charikar et al.~\cite{CLLPPSS05}
  from grammars to attractors. We show that to approximate
  \textsc{Minimum-$k$-VertexCover} (minimum vertex cover for graphs
  with vertex-degree bounded by $k$) in polynomial time below a factor
  $1+\epsilon$, for any constant $\epsilon>0$, it suffices to
  approximate \textsc{Minimum-$k$-Attractor} in polynomial time below
  a factor $1+\delta$, where $\delta=\epsilon/(2k^3+4k^2-3k+1)$. In
  other words, if \textsc{Minimum}-$k$-\textsc{Attractor} has a PTAS
  then \textsc{Minimum}-$k$-\textsc{VertexCover} also has a
  PTAS. Since \textsc{Minimum-$k$-VertexCover} is
  APX-hard~\cite{papa88}, this will yield the claim.
	
  Let $G=(V, E)$ be an undirected graph with vertex-degree bounded by
  $k$.  Assume that $|V|\leq|E|$ and that $G$ has no isolated
  vertices.\footnote{\textsc{Minimum}-$k$-\textsc{VertexCover} is
    still APX-hard under this assumption, since a PTAS for this case
    would give us a PTAS for the general case.}  Let $\mathcal{U}_G=E$
  and $C_G=\{E_v \mid v\in V\}$, where $E_v=\{e\in E \mid e\text{ is
    incident to } v\}$. Then, the size of the minimum $k$-set cover
  for $C_G$ is $p$ if and only if the minimum $k$-vertex cover of $G$
  has size $p$.  Consider the string $T_G := T_{C_G}$ as in
  Theorem~\ref{thm:k-attractor-npc}. The smallest $k$-attractor of
  $T_G$ has size $(8+4k)|E|-3|V|+p$ (since the universe size is $|E|$,
  the number of sets in $C_G$ is $|V|$, and their total cardinality is
  $2|E|$) if and only if the smallest vertex cover of $G$ has size
  $p$.

  Assume it is NP-hard to approximate \textsc{Minimum-$k$-VertexCover}
  below the ratio $1+\epsilon$. Then it is also NP-hard to approximate
  the smallest $k$-attractor for $T_G$ below the ratio
  \[
    r = \frac{(8+4k)|E|-3|V|+(1+\epsilon)p}{(8+4k)|E|-3|V|+p} = 1 +
    \frac{\epsilon p}{(8+4k)|E|-3|V|+p}.
  \]
	
  Since all vertices have degree at most $k$, $2|E|\leq k|V|$.
  Furthermore, since each vertex can cover at most $k$ edges, the size
  of the minimum $k$-vertex cover, $p$, must be at least
  $\frac{1}{k}|E|\geq \frac{1}{k}|V|$.  The expression above achieves
  its minimum when $|E|$ is large and $p$ is small.  From the
  constraints $|E| \leq\frac{k}{2}|V|$ and $p\geq \frac{1}{k}|V|$, we
  thus get the lower bound
  \[
    r \geq 1 + \frac{\epsilon \cdot \frac{1}{k}|V|}{(8+4k)\cdot
      \frac{k}{2}|V| - 3|V| + \frac{1}{k}|V|} = 1+
    \frac{\epsilon}{2k^3 + 4k^2 - 3k + 1}.
  \]
\end{proof}

\begin{corollary}
  \label{thm:approx-hardness-constant}
  For every constant $\epsilon>0$ and every (not necessarily constant)
  $k \geq 3$, it is NP-hard to approximate
  \textsc{Minimum-$k$-Attractor} within factor $11809/11808-\epsilon$.
\end{corollary}
\begin{proof}
  By~\cite{bk99}, \textsc{Minimum-$3$-VertexCover} is NP-hard to
  approximate below a factor $1+\epsilon_3=\frac{145}{144}$. By
  Theorem~\ref{thm:apx-complete} it is equally hard to approximate
  \textsc{Minimum-$3$-Attractor} below $1+\frac{\epsilon_3}{2k^3 +
    4k^2 - 3k + 1}$, where $k=3$.  The claim for larger (and
  non-constant) $k$ follows from the property
  $\gamma_{k}^*(T_G)=\gamma_{k'}^*(T_G)$, $k < k' \leq |T_G|$ of the
  string $T_G$ used in the proof on Theorem~\ref{thm:apx-complete}.
\end{proof}

Theorem~\ref{thm:apx-complete} implies a $2k$-approximation algorithm
for \textsc{Minimum-$k$-Attractor}, $k\geq 3$.  By reducing the
problem to \textsc{Minimum-$k$-SetCover} we can however obtain a
better ratio.

\begin{theorem}
  \label{thm:log-approx}
  For any $k\geq 3$, \textsc{Minimum-$k$-Attractor} can be
  approximated in polynomial time up to a factor of
  $\mathcal{H}(k(k+1)/2)$, where
  $\mathcal{H}(p)=\sum_{i=1}^{p}\frac{1}{i}$ is the $p$-th harmonic
  number.  In particular, \textsc{MinimumAttractor} can be
  approximated to a factor $\mathcal{H}(n(n+1)/2)\leq 2\ln((n+1)/\sqrt
  2) + 1$.
\end{theorem}
\begin{proof}
  We first show that in polynomial time we can reduce
  \textsc{Minimum-$k$-Attractor} to an instance of
  \textsc{Minimum-$k'$-SetCover} for $k'=k(k+1)/2$.  Let $T$ be the
  input string of length $n$. Consider the set $\mathcal{U}$ of all
  distinct substrings of $T$ of length $\leq k$. The size of
  $\mathcal{U}$ is at most $kn$, i.e., polynomial in $n$. We create a
  collection $C$ of sets over $\mathcal{U}$ as follows. For any
  position $i\in[1,n]$ in $T$ take all distinct substrings of length
  $\leq k$ that have an occurrence containing position $i$ (there is
  at most $p$ such substrings of length $p$ and hence not more than
  $k(k+1)/2$ in total) and add a set containing those substrings to
  $C$. It is easy to see that \textsc{Minimum-$k$-attractor} for $T$
  has the same size as \textsc{Minimum-$k'$-SetCover} for $C$. Since
  the latter can be approximated to a factor
  $\mathcal{H}(k')$~\cite{j73}, the claim follows.
\end{proof}

For constant $k\geq 3$, Duh and F\"{u}rer~\cite{df97} describe an
approximation algorithm based on semi-local optimization that achieves
an approximation ratio of $\mathcal{H}(k)-1/2$ for
\textsc{Minimum-$k$-SetCover}. Thus, we obtain an improved
approximation ratio for constant $k$.

\begin{theorem}
  Let $\mathcal{H}(p)=\sum_{i=1}^{p}\frac{1}{i}$ be the $p$-th
  harmonic number.
  \label{thm:better-approx}
  For any constant $k\geq 3$, \textsc{Minimum-$k$-Attractor} can be
  approximated in polynomial time up to a factor of
  $\mathcal{H}(k(k+1)/2)-1/2$.
\end{theorem}

\section{Optimal-Time Random Access}

In this section we show that the simple string attractor property
introduced in Definition \ref{def: string attractor} is sufficient to
support random access in \emph{optimal} time on string attractors and,
in particular, on most dictionary-compression schemes.  We show this
fact by extending an existing lower bound of Verbin and
Yu~\cite{CVY13} (holding on grammars) and by providing a data
structure matching this lower bound. First, we reiterate the main step
of the proof in~\cite{CVY13}, with minute technical details tailored
to our needs.

\begin{theorem} [Verbin and Yu~\cite{CVY13}]\label{th:low_bound}
  Let $g$ be the size of any SLP for a string of length $n$.  Any
  static data structure taking $\bigO(g\;{\rm polylog}\;n)$ space
  needs $\Omega(\log n / \log \log n)$ time to answer random access
  queries.
\end{theorem}
\begin{proof}
  Consider the following problem: given $m$ points on a grid of size
  $m \times m^{\epsilon}$, where $\epsilon>0$ is some constant, build
  a data structure answering 2-sided parity range-counting queries,
  i.e., for any position $(x,y)$ find the number (modulo 2) of points
  with coordinates in $[1,x]\times[1,y]$. Any static data structure
  answering such queries using $\bigO(m\;{\rm polylog}\;m)$ words of
  space must have a query time of $\Omega(\log m / \log \log
  m)$~\cite[Lem. 5]{CVY13}.  Assume that our claim does not hold,
  i.e., for any SLP of size $g$, there exists a static data structure
  $D$ of size $\bigO(g\;{\rm polylog}\;n)$ that answers access queries
  in $o(\log n / \log \log n)$ time. Now take any instance of the
  above range-counting problem, i.e., a set of $m$ points on a
  grid. Take the string of length $n=m^{1+\epsilon}$ encoding answers
  to all possible queries (call it the \emph{answer string}) in
  row-major order. This string, by~\cite[Lem. 6]{CVY13}, has an SLP of
  size $g\in \bigO(m \log m)$.  Thus, $D$ takes $\bigO(g\;{\rm
    polylog}\;n) =\bigO(m\;{\rm polylog}\;m)$ space and answers access
  (and hence also range-counting) queries in $o(\log n / \log \log
  n)=o(\log m / \log \log m)$ time,
  contradicting~\cite[Lem. 5]{CVY13}.
\end{proof}

The key observation for extending the above lower bound to other
compression schemes and to string attractors is that we can use known
reductions from SLPs to obtain a different representation (e.g., a
collage system or a macro scheme) of size at most $g$. For example,
the fact that $z\leq g^*$~\cite{rytter2003application} immediately
implies that the above bound also holds within $\bigO(z\;{\rm
  polylog}\;n)$ space. Hence, for any compression method that is at
least as powerful as SLPs we can generalize the lower bound.

\begin{theorem}\label{thm:low_bound_compressors}
  Let $T\in\Sigma^n$ and let $\alpha$ be any of these measures:
  \begin{enumerate}
    \item[(1)] the size $\gamma$ of a string attractor for $T$,
    \item[(2)] the size $g_{rl}$ of an RLSLP for $T$,
    \item[(3)] the size $c$ of a collage system for $T$,
    \item[(4)] the size $z$ of the LZ77 parse of $T$,
    \item[(5)] the size $b$ of a macro scheme for $T$.
  \end{enumerate}
  Then, $\Omega(\log n/\log\log n)$ time is needed to access one
  random position of $T$ within $\bigO(\alpha\ \mathrm{polylog}\ n)$
  space.
\end{theorem}
\begin{proof}
  Let $G$ be the SLP of size $g$ used in Theorem \ref{th:low_bound} to
  compress the answer string. By our reduction stated in Theorem
  \ref{th:attr_c}, we can build a string attractor of size $\gamma
  \leq g$, therefore $\gamma\ \mathrm{polylog}\ n
  \in\bigO(g\ \mathrm{polylog}\ n)$ and bound (1) holds. Since RLSLPs
  and collage systems are extensions of SLPs, $G$ is also an RLSLP and
  a collage system for $T$, hence bounds (2) and (3) hold
  trivially. From the relation $z\leq
  g^*$~\cite{rytter2003application} we have that
  $z\ \mathrm{polylog}\ n \in\bigO(g\ \mathrm{polylog}\ n)$, therefore
  bound (4) holds. Finally, LZ77 is a particular unidirectional parse,
  and macro schemes are extensions of unidirectional parses, hence
  bound (5) holds.
\end{proof}

We now describe a parametrized data structure based on string
attractors matching lower bounds (1-5) of Theorem
\ref{thm:low_bound_compressors}. Our result generalizes Block
Trees~\cite{BGGKOPT15} (where blocks are only copied left-to-right)
and a data structure proposed very recently by Gagie et
al.~\cite{gagie2017optimal} supporting random access on the RLBWT
(where only constant out-degree is considered).

\begin{theorem}\label{thm:extract}
  Let $T[1..n]$ be a string over alphabet $[1..\sigma]$, and let
  $\Gamma$ be a string attractor of size $\gamma$ for $T$. For any
  integer parameter $\tau\geq 2$, we can store a data structure of
  $\bigO(\gamma \tau \log_\tau(n/\gamma))$ $w$-bit words supporting
  the extraction of any length-$\ell$ substring of $T$ in
  $\bigO(\log_\tau(n/\gamma) + \ell\log(\sigma)/w)$ time.
\end{theorem}
\begin{proof}
  We describe a data structure supporting the extraction of $\alpha =
  \frac{w\log_\tau(n/\gamma)}{\log\sigma}$ packed characters in
  $\bigO(\log_\tau(n/\gamma))$ time. To extract a substring of length
  $\ell$ we divide it into $\lceil\ell/\alpha\rceil$ blocks and
  extract each block with the proposed data structure. Overall, this
  will take $\bigO((\ell/\alpha+1)\log_\tau(n/\gamma))
  =\bigO(\log_\tau(n/\gamma) + \ell\log(\sigma)/w)$ time.
	
  Our data structure is organized into $\bigO(\log_\tau(n/\gamma))$
  levels. For simplicity, we assume that $\gamma$ divides $n$ and that
  $n/\gamma$ is a power of $\tau$.  The top level (level 0) is
  special: we divide the string into $\gamma$ blocks
  $T[1..n/\gamma]\,T[n/\gamma+1..2n/\gamma]\dots T[n-n/\gamma+1..n]$
  of size $n/\gamma$.  Intuitively, at each level $i>0$ we associate
  to each $j\in \Gamma$ two \emph{context substrings} of length $s_i =
  n/(\gamma\cdot \tau^{i-1})$ flanking position $j$. These substrings
  are divided in a certain number of (overlapping) blocks of length
  $s_{i}/\tau = s_{i+1}$. Each block is then associated to an
  occurrence at level $i+1$ overlapping some element $j'\in\Gamma$
  (possible by definition of $\Gamma$). At some particular level $i^*$
  (read the formal description below) we store explicitly all
  characters in the context substrings. To extract a substring of
  length $\alpha$, we will map it from level $0$ to level $i^*$, and
  then extract naively using the explicitly stored characters.

  More formally, for levels $i>0$ and for every element $j\in\Gamma$,
  we consider the $2\tau$ non-overlapping blocks of length $s_{i+1}$
  forming the two context substrings flanking $j$: $T[j-s_{i+1}\cdot k
    +1 ... j-s_{i+1}\cdot (k-1)]$ and $T[j+s_{i+1}\cdot (k-1) + 1
    ... j+s_{i+1}\cdot k]$, for $k=1,\dots, \tau$.  We moreover
  consider a sequence of $2\tau-1$ additional consecutive and
  non-overlapping blocks of length $s_{i+1}$, starting in the middle
  of the first block above defined and ending in the middle of the
  last: $T[j-s_{i+1}\cdot k + 1 + s_{i+1}/2 ... j-s_{i+1}\cdot (k-1)+
    s_{i+1}/2]$ for $k=1,\dots, \tau$, and $T[j+s_{i+1}\cdot (k-1) + 1
    + s_{i+1}/2 ... j+s_{i+1}\cdot k + s_{i+1}/2]$, for $k=1,\dots,
  \tau-1$.  Note that, with this choice of blocks, at level $i$ for
  any substring $S$ of length at most $s_{i+1}/2$ inside the context
  substrings around elements of $\Gamma$ we can always find a block
  fully containing $S$. This property will now be used to map
  ``short'' strings from the first to last level of our structure
  without splitting them, until reaching explicitly stored characters
  at some level $i^*$ (see below).
	
  From the definition of string attractor, blocks at level $0$ and
  each block at level $i > 0$ have an occurrence at level $i+1$
  crossing some position in $\Gamma$. Such an occurrence can be fully
  identified by the coordinate $\langle \mathit{off}, j\rangle$, for
  $0 \leq \mathit{off} < s_{i+1}$ and $j\in\Gamma$, indicating that
  the occurrence starts at position $j - \mathit{off}$.  Let $i^*$ be
  the smallest number such that $s_{i^*+1} < 2\alpha =
  \frac{2w\log_\tau(n/\gamma)}{\log\sigma}$. Then $i^*$ is the last
  level of our structure.  At this level, we explicitly store a packed
  string with the characters of the blocks. This uses in total
  $\bigO(\gamma \cdot s_{i^*}\log(\sigma)/w)
  =\bigO(\gamma\tau\log_\tau(n/\gamma))$ words of space. All the
  blocks at levels $0\leq i < i^*$ store instead the coordinates
  $\langle \mathit{off},j\rangle$ of their primary occurrence in the
  next level. At level $i^*-1$, these coordinates point inside the
  strings of explicitly stored characters.
	
  Let $S= T[i..i+\alpha-1]$ be the substring to be extracted. Note
  that we can assume $n/\gamma \geq \alpha$; otherwise the whole
  string can be stored in plain packed form using $n\log(\sigma)/w <
  \alpha \gamma\log(\sigma)/w \in\bigO(\gamma\log_\tau (n/\gamma))$
  words and we do not need any data structure. It follows that $S$
  either spans two blocks at level 0, or it is contained in a single
  block. The former case can be solved with two queries of the latter,
  so we assume, without losing generality, that $S$ is fully contained
  inside a block at level $0$. To retrieve $S$, we map it down to the
  next levels (using the stored coordinates) as a contiguous substring
  as long as this is possible, that is, as long as it fits inside a
  single block. Note that, thanks to the way blocks overlap, this is
  always possible as long as level $i$ is such that $\alpha \leq
  s_{i+1}/2$. By definition, then, we arrive in this way precisely to
  level $i^*$, where characters are stored explicitly and we can
  return the packed substring. Note also that, since blocks in the
  same level have the same length, at each level we spend only
  constant time to find the pointer to the next level (this requires a
  simple integer division).
\end{proof}

Table \ref{table:extract A-DAG} reports some interesting space-time
trade-offs achievable with our data structure.  For
$\tau=\log^\epsilon n$, the data structure takes
$\bigO(\gamma\ \mathrm{polylog}\ n)$ space and answers random access
queries in $\bigO(\log(n/\gamma)/\log\log n)$ time, which is optimal
by Theorem \ref{thm:low_bound_compressors} (note that
$\log(n/\gamma)\in \Theta(\log n)$ for the string used in Theorem
\ref{thm:low_bound_compressors}, so the structure does not break the
lower bound). Choosing $\tau = (n/\gamma)^\epsilon$, space increases
to $\bigO(\gamma^{1-\epsilon}n^\epsilon)$ words and query time is
optimal \emph{in the packed setting}.  Note that our data structure is
\emph{universal}: given any dictionary-compressed representation, by
the reductions of Section \ref{sec:compressors->SA} we can derive a
string attractor of the same asymptotic size and build our data
structure on top of it. By Theorems \ref{th:low_bound} and
\ref{thm:low_bound_compressors} we obtain:

\begin{corollary}
  For $\tau=\log^\epsilon n$ (for any constant $\epsilon>0$), the data
  structure of Theorem \ref{thm:extract} supports random access in
  optimal time on string attractors, SLPs, RLSLPs, LZ77, collage
  systems, and macro schemes.
\end{corollary}

\section{Conclusions}

In this paper we have proposed a new theory unifying all known
dictionary compression techniques.  The new combinatorial object at
the core of this theory --- the string attractor --- is NP-hard to
optimize within some constant in polynomial time, but logarithmic
approximations can be achieved using compression algorithms and
reductions to well-studied combinatorial problems.  We have moreover
shown a data structure supporting optimal random access queries on
string attractors and on most known dictionary compressors. Random
access stands at the core of most compressed computation techniques;
our results suggest that compressed computation can be performed
independently of the underlying compression scheme (and even in
optimal time for some queries).

An interesting view for future research is to treat (the size of the
smallest) $k$-attractors as a measure of string compressibility akin
to the $k$-th order empirical entropy (which has proven to be an
accurate and robust measure for texts that are not highly-repetitive),
as it exhibits a similar regularity, e.g., $\gamma_k^*(X) \leq
\gamma_{k+1}^*(X)$ for any $k$, while being sensitive to repetition:
$\gamma_{k}^{*}(X^t) \leq \gamma_{k}^{*}(X)+1$.

Another use of our techniques could be to provide a linear ordering of
compression algorithms based on how well they approximate the smallest
attractor.  For example, the unary string shows that a ``weak''
compression like LZ78 in the worst case cannot achieve a better ratio
than $|LZ78|/\gamma^* \in \Omega(\sqrt{n})$, while we showed that LZ77
achieves (via our reductions from attractors) $|LZ77|/\gamma^* \in
\bigO({\rm polylog}\;n)$ ratio. Relatedly, it is still an open problem
to determine whether the smallest attractor can be approximated up to
$o(\log n)$ ratio in polynomial time for all strings.  Even within
logarithmic ratio, we have left open the problem of efficiently
computing such an approximation. A naive implementation of our
algorithm based on set-cover runs in cubic time.

It would also be interesting to further explore the landscape of
compressed data structures based on string attractors. In this paper
we showed that the simple string attractor property is sufficient to
support random access. Is this true for more complex queries such as,
e.g., indexing?

Finally, an intriguing problem is that of optimal approximation of
string $k$-attractors; e.g., what is the complexity of the 2-attractor
problem? what is, assuming P$\neq$NP, the best approximation ratio for
the minimum 3-attractor problem? For the latter question, in this
paper we gave a lower bound of 11809/11808
(Corollary~\ref{thm:approx-hardness-constant}) and an upper bound of
$1.95$ (Theorem~\ref{thm:better-approx}).

\section*{Acknowledgments}

We received useful suggestions from many people during the write-up
of this paper. We would like to thank (in alphabetical order) Philip
Bille, Anders Roy Christiansen, Mikko Berggren Ettienne, Travis
Gagie, Inge Li G\o rtz, Juha K{\"a}rkk{\"a}inen, Gonzalo Navarro,
Alberto Policriti, and Esko Ukkonen for the great feedback.

This work was partially funded by Danish Research Council grant
DFF-4005-00267 and by the project \linebreak MIUR-SIR CMACBioSeq (``Combinatorial
methods for analysis and compression of
biological sequences'') grant RBSI146R5L.

\bibliographystyle{plainurl}
\bibliography{stoc2018}

\end{document}